%% file: main.tex
\documentclass{llncs}

\usepackage[english]{babel}
\usepackage[T1]{fontenc}
\usepackage[latin1]{inputenc}
\usepackage{amssymb,amsmath,amscd,latexsym,graphicx}

\usepackage{gastex} 
\usepackage[usenames]{color}

\usepackage{fullpage}

\input{macros.tex}

\begin{document}
\title{Infinite-state games with finitary conditions}

\author{
Krishnendu Chatterjee \inst{1} \and
Nathana\"el Fijalkow \inst{2}
}
\institute{IST Austria, Klosterneuburg, Austria \\
\email{krishnendu.chatterjee@ist.ac.at}
\and
LIAFA, CNRS \& Universit\'e Denis Diderot - Paris 7, France \\
Institute of Informatics, University of Warsaw, Poland \\
\email{nath@liafa.univ-paris-diderot.fr}
}

\maketitle

\begin{abstract}
We study two-player zero-sum games over infinite-state graphs with boundedness conditions.

Our first contribution is about the strategy complexity, \textit{i.e} the memory 
required for winning strategies: we prove that over general infinite-state graphs,
memoryless strategies are sufficient for finitary B\"uchi games, 
and finite-memory suffices for finitary parity games.

We then study pushdown boundedness games, with two contributions.
First we prove a collapse result for pushdown $\omega B$ games, implying the decidability of solving these games.
Second we consider pushdown games with finitary parity along with stack boundedness conditions,
and show that solving these games is $\EXPTIME$-complete.
\end{abstract}

\section{Introduction}

\input{intro}

\section{Definitions}

\input{definitions}

\section{Strategy complexity for finitary conditions over infinite-state games}
\label{sec:strategy_complexity}

In this section we give characterizations of the winning regions
for finitary conditions over infinite arenas,
and use them to establish the strategy complexity for both players.
The main results are summarized in the following theorem.

\begin{theorem}[Strategy complexity for finitary games]
\label{thm:strategy_complexity}
The following assertions hold:
\begin{enumerate}
	\item For all finitary B\"uchi games, Eve has a memoryless winning strategy 
from her winning set.
	\item For all finitary parity games, Eve has a finite-memory winning strategy from her winning set 
that uses at most $\dhalf$ memory states, where $\ell$ is the number of odd colors.
\end{enumerate}
\end{theorem}

\input{strategy_complexity}

\section{Pushdown $\omega B$ games}
\label{sec:pushdown}

In this section we consider pushdown $\omega B$ games and prove a collapse result.
Along with previous results~\cite{BlumensathColcombetKuperbergVandenboom13,Bojanczyk04},
this implies that determining the winner in such games is decidable.

\input{games}

\subsection{Regular sets of configurations and alternating $\PP$-automata}
\input{p_automata}

\subsection{The collapse result}
\input{collapse}

\subsection{Decidability of pushdown $\omega B$ games}
\input{decidability}

\subsection{Lower bound on the collapse for finitary conditions}
\input{lower_bound}

\section{Pushdown games with finitary and stack boundedness conditions}
\label{sec:stack_boundedness}
In this section, we consider pushdown games with finitary parity along with stack boundedness conditions,
following~\cite{BouquetSerreWalukiewicz03,Gimbert04}.
We prove that solving such games is $\EXPTIME$-complete.
This is achieved by a reduction which relies on two ideas,
that we present separately; the first is a reduction
from finitary parity to bounded parity, and the second
a collapse result for finitary B\"uchi along with stack boundedness conditions.
We then show how to combine them to obtain a complete reduction,
with an optimal complexity.

We denote by $\SB$ the stack boundedness condition:
$$\SB = \set{\pi \mid \exists N, 
\begin{array}{c}
\textrm{ all configurations in } \pi \textrm{ have}\\
\textrm{ stack height less than } N
\end{array}
}\ .$$

\subsection{A reduction from finitary parity to bounded parity}
\input{reduction}

\subsection{The special case of B\"uchi conditions}
\input{buchi_stack_boundedness}

\subsection{The complete reduction}
\input{parity_stack_boundedness}

\medskip\noindent\textbf{Conclusion.}
We studied boundedness games over infinite arenas, and investigated two questions. 
First, the strategy complexity over general infinite arenas; we proved
that finite-memory winning strategies exist for finitary parity games.
It remains open to extend this to cost-parity games~\cite{FijalkowZimmermann12}.
Second, the decidability of pushdown games; we proved that 
pushdown $\omega B$-games are decidable, and pushdown games with finitary parity
along with stack boundedness conditions are $\EXPTIME$-complete.

\medskip\noindent\textbf{Acknowledgments.}
We thank Denis Kuperberg and Thomas Colcombet for sharing and explaining~\cite{BlumensathColcombetKuperbergVandenboom13},
Damian Niwinski for raising the question of pushdown finitary games,
Olivier Serre for many inspiring discussions and Florian Horn for interesting suggestions.
We are grateful to the LICS anonymous reviewers for their valuable comments.

\bibliographystyle{plain}
\bibliography{bib}

\end{document}

%% file: macros.tex
\newcommand{\set}[1]{\{ #1 \}}

\newcommand{\dhalf}{\ell + 1}

\newcommand{\infi}{\mathrm{Inf}}
\newcommand{\distk}{\mathrm{dist}_k}
\newcommand{\Pre}{\mathrm{Pre}}
\newcommand{\Act}{\mathrm{Act}}
\newcommand{\SB}{\mathrm{BndSt}}

\newcommand{\PP}{\mathcal{P}}

\newcommand{\Attr}{\mathrm{Attr}}
\newcommand{\Attreve}{\Attr^{\mathrm{E}}}
\newcommand{\Attrevef}{\Attr^{\mathrm{E}}(F)}

\newcommand{\safety}{\mathrm{Safety}}
\newcommand{\bounded}{\mathrm{Bounded}}
\newcommand{\limitbounded}{\mathrm{LimitBounded}}

\newcommand{\buc}{\mathrm{B\ddot{u}chi}}
\newcommand{\bucf}{\buc(F)}
\newcommand{\bucfn}{\buc(F,N)}

\newcommand{\parity}{\mathrm{Parity}}
\newcommand{\parp}{\parity(c)}

\newcommand{\fin}{\mathrm{Fin}}
\newcommand{\bnd}{\mathrm{Bnd}}
\newcommand{\co}{\mathrm{Co}}

\newcommand{\push}{\mathrm{push}}
\newcommand{\pop}{\mathrm{pop}}
\newcommand{\Skip}{\mathrm{skip}}

\newcommand{\EXPTIME}{\mathrm{EXPTIME}}

\newcommand{\N}{\mathbb{N}}

\newcommand{\play}{\pi}
\newcommand{\G}{\mathcal{G}}

\newcommand{\VE}{V_{E}}
\newcommand{\VA}{V_{A}}
\newcommand{\W}{\mathcal{W}}
\newcommand{\WE}{\W_{E}}
\newcommand{\WA}{\W_{A}}
\newcommand{\nex}{\nu}
\newcommand{\M}{\mathcal{M}}
\newcommand{\up}{\mu}
\newcommand{\A}{\mathcal{A}}
\newcommand{\B}{\mathcal{B}}

\def\qed{\rule{0.4em}{1.4ex}}

\newcommand{\Bn}{\mathbb{B}}
\newcommand{\MSO}{\textrm{MSO}}

%% file: intro.tex
\noindent{\bf Games on graphs.} 
Two-player games played on graphs is a powerful mathematical framework 
to analyze several problems 
in computer science as well as mathematics.
In particular, when the vertices of the graph represent the states 
of a reactive system and the edges represent the transitions, then the synthesis problem
(Church's problem) asks for the construction of a winning strategy in a 
game played on the graph~\cite{BuchiLandweber69,McNaughton93}.
Game-theoretic formulations have also proved useful for the 
verification, refinement, and compatibility 
checking of reactive systems~\cite{AlurHenzingerKupferman02}; 
and has deep connection with automata theory and logic, \textit{e.g} 
the celebrated decidability result
of monadic second-order logic over infinite trees due to Rabin~\cite{Rabin69}.
%The vertex set of the graph is partitioned into vertices controlled by Eve and vertices controlled by Adam.
%A token is initially placed at an initial vertex, and the player who controls 
%the current vertex moves the token along an edge.
%This process is repeated forever, and gives rise to an outcome of the game, 
%called a {\em play}, that consists of the infinite sequence of states that are visited.

\smallskip\noindent{\bf Omega-regular conditions: strengths and weaknesses.} 
In the literature, two-player games on finite-state graphs with $\omega$-regular 
conditions have been extensively studied~\cite{EmersonJutla88,EmersonJutla91,LNCS2500,GurevichHarrington82,Zielonka98}. 
The class of $\omega$-regular languages provides a robust specification language 
for solving control and verification problems (see, \textit{e.g}, \cite{PnueliRosner89}).
Every $\omega$-regular condition
can be decomposed into a safety part and a liveness part~\cite{AlpernSchneider85}.
The safety part ensures that the component will not do anything ``bad''
(such as violate an invariant) within any finite number of transitions.
The liveness part ensures that the component will do something ``good''
(such as proceed, or respond, or terminate) in the long-run.
Liveness can be violated only in the limit, by infinite sequences of 
transitions, as no bound is stipulated on when the ``good'' thing must 
happen.
This infinitary, classical formulation of liveness has both strengths and 
weaknesses. 
A main strength is robustness, in particular, independence from the chosen 
granularity of transitions.
Another important strength is simplicity, allowing liveness to serve as an 
abstraction for complicated safety conditions.
For example, a component may always respond in a number of transitions 
that depends, in some complicated manner, on the exact size of the 
stimulus.
Yet for correctness, we may be interested only that the component will 
respond ``eventually''.
However, these  strengths also point to a weakness of the classical 
definition of liveness:
it can be satisfied by components that in practice are quite
unsatisfactory because no bound can be put on their response time.

\smallskip\noindent{\bf Stronger notion of liveness: finitary conditions.}
For the weakness of the infinitary formulation of liveness, alternative and 
stronger formulations of liveness have been proposed.
One of these is {\em finitary} liveness~\cite{AlurHenzinger98}:
it is satisfied if \textit{there exists} a bound~$N$ such that every stimulus is followed by a response 
within $N$ transitions.
Note that it does not insist on a response within a known bound~$N$
(\textit{i.e}, every stimulus is followed by a response within $N$ transitions), 
but on response within some unknown bound, which can be arbitrarily large;
in other words, the response time must not grow forever from one stimulus to the next.
In this way, finitary liveness still maintains the robustness (independence 
of step granularity) and simplicity (abstraction of complicated safety conditions) of 
traditional liveness, while removing unsatisfactory implementations.

All $\omega$-regular languages can be defined by a deterministic parity
automaton; the parity condition assigns to each state an integer representing a priority,
and requires that in the limit, every odd priority is followed by a lower even priority.
Its finitary counterpart, the finitary parity condition, 
strengthens this by requiring the existence of a bound $N$ such that 
in the limit every odd priority is followed by a lower even priority within $N$ transitions.

\smallskip\noindent{\bf Bounds in $\omega$-regularity.}
The finitary conditions are closely related to the line of work initiated by Boja{\'n}czyk in~\cite{Bojanczyk04},
where the $\MSO + \Bn$ logic was defined, generalizing $\MSO$ by adding a bounding quantifier $\Bn$.
The satisfiability problem for this logic has been deeply investigated 
(see for instance~\cite{Bojanczyk04,BojanczykColcombet06,BojanczykTorunczyk12}),
but the decidability for the general case is still open.
A fragment of $\MSO + \Bn$ over infinite words was shown to be decidable in~\cite{BojanczykColcombet06}, 
by introducing the model of $\omega B$-automata, which manipulate counters.
They perform three kind of actions on counters: increment ($i$),
reset ($r$) or nothing ($\varepsilon$).
The relation with finitary conditions has been investigated in~\cite{ChatterjeeFijalkow11},
where it was shown that automata with finitary conditions exactly correspond to star-free $\omega B$-expressions.
Moreover, the finitary conditions are recognized by $\omega B$-automata,
hence they can be considered as a subcase of $\omega B$-conditions.

\smallskip\noindent{\bf Regular cost-functions.}
A different perspective for bounds in $\omega$-regularity
was developed by Colcombet in~\cite{Colcombet09} with functions instead of 
languages, giving rise to the theory of regular cost-functions and cost-$\MSO$.
The decidability of cost-$\MSO$ over finite trees was established in~\cite{ColcombetLoeding10},
but its extension over infinite trees is still open,
and would imply the decidability of the index of the non-deterministic Mostowski hierarchy~\cite{ColcombetLoeding08},
a problem open for decades.
A subclass of cost-$\MSO$ called temporal cost logic was introduced in~\cite{ColcombetKuperbergLombardy10}
and is the counterpart of finitary conditions for regular cost-functions~\cite{ChatterjeeFijalkow11},
also reminiscent of desert automata~\cite{Kirsten04}.

\smallskip\noindent{\bf Quantification order.}
The essential difference between the approaches underlying the logics $\MSO + \Bn$ and cost-$\MSO$ is a quantifier switch.
We illustrate this in the context of games: a typical property expressed in $\MSO + \Bn$
is ``there exists a strategy, such that for all plays, there exists a bound on the counter values'',
while cost-$\MSO$ allows to express properties like 
``there exists a strategy, there exists a bound $N$, such that for all plays, the counter values are bounded by $N$''.
In other words, $\MSO + \Bn$ expresses non-uniform bounds while bounds in cost-$\MSO$ are uniform.

\smallskip\noindent{\bf Solving boundedness games.}
Games over finite graphs with finitary conditions have been studied in~\cite{ChatterjeeHenzingerHorn09},
leading to very efficient algorithms: finitary parity games can be solved in polynomial time
(unlike classical parity games).
In this paper, we study games over infinite graphs with finitary conditions,
and then focus on the widely studied class of pushdown games,
which model sequential programs with recursion.
This line of work belongs to the tradition of infinite-state systems and games
(see \textit{e.g}~\cite{AbdullaBouajjanidOrso08,BrazdilJAncarKucera10}).
Pushdown games with the classical reachability and parity conditions have 
been studied in~\cite{AlurTorreMadhusudan06,Walukiewicz01}.
It has been established in~\cite{Walukiewicz01} that the problem of deciding the 
winner in pushdown parity games is EXPTIME-complete. 
However, little is known about pushdown games with boundedness conditions; 
one notable exception is parity and stack boundedness conditions~\cite{BouquetSerreWalukiewicz03,Gimbert04}.
The stack boundedness condition naturally arises
with the synthesis problem in mind,
since bounding the stack amounts to control
the depth of recursion calls of the sequential program.

\smallskip\noindent{\bf Memoryless determinacy for infinite-state games.}
Our motivation to prove the existence of finite-memory winning strategies is towards
automata theory, where several constructions rely on the existence 
of memoryless winning strategies (for parity games):
for instance to complement tree automata~\cite{EmersonJutla91},
or to simulate alternating two-way tree automata by non-deterministic ones~\cite{Vardi98}.

In particular, Colcombet pointed out in~\cite{Colcombet13} that the remaining difficulty to establish the decidability
of cost-$\MSO$ over infinite trees is a good understanding of boundedness games,
and more specifically the cornerstone is to extend the memoryless determinacy of parity games over infinite graphs,
following~\cite{EmersonJutla88,EmersonJutla91,GurevichHarrington82}.

\smallskip\noindent{\bf Our contributions.}
We study two questions about infinite-state games with boundedness conditions:
the \emph{memory requirements} of winning strategies
and the \emph{decidability} of solving a pushdown game. 

\smallskip\noindent{\it Strategy complexity.}
We give (non-effective) characterizations of the winning regions
for finitary games over countably infinite graphs,
implying a complete picture of the strategy complexity.
Most importantly, we show that for finitary B\"uchi games
memoryless strategies suffice, and 
that for finitary parity games, memory of size $\dhalf$ suffices, 
where $\ell$ is the number of odd priorities in the parity condition.

\smallskip\noindent{\it Pushdown games.} 
We present two contributions.

First we consider pushdown boundedness games
and prove that the following statements are equivalent:
``there exists a strategy, such that for all plays, there exists a bound on the counter values 
and the parity condition is satisfied'' and 
``there exists a strategy, there exists a bound $N$, such that for all plays, \textit{eventually} 
the counter values are bounded by $N$ and the parity condition is satisfied''.
We refer to this as a collapse result, as it reduces a quantification with non-uniform bounds (in the fashion of $\MSO + \Bn$)
to one with uniform bounds (\`a la cost-$\MSO$).
Using this, we obtain the decidability of solving such games
relying on previous results~\cite{BlumensathColcombetKuperbergVandenboom13,Bojanczyk04}.

Second we consider pushdown games with finitary parity along with stack boundedness conditions,
and establish that solving these games is $\EXPTIME$-complete.

%% file: definitions.tex
\medskip\noindent{\bf Arenas and games.}
The games we consider are played on an \emph{arena} $\A = (V,(\VE,\VA),E)$, 
which consists of a (potentially infinite but countable) graph $(V,E)$
and a partition $(\VE,\VA)$ of the vertex set $V$.
A vertex is controlled by Eve and depicted by a circle if it belongs to $\VE$
and controlled by Adam and depicted by a square if it belongs to $\VA$.
Playing consists in moving a pebble along the edges: initially placed on a vertex $v_0$,
the pebble is sent along an edge chosen by the player who controls the vertex.
From this interaction results a path in the graph, called a \emph{play} 
and usually denoted $\pi = v_0, v_1, \ldots$.
To avoid the nuisance of dealing with finite plays, we assume that the graphs have no dead-ends:
all vertices have an outgoing edge,
so the plays are infinite.
We denote by $\Pi$ the set of all plays, and define \emph{conditions}
for a player by sets of winning plays $\Omega \subseteq \Pi$. 
The games are zero-sum, which means that if Eve's condition is $\Omega$, then Adam's condition is
$\Pi \setminus \Omega$, usually denoted by ``$\co \Omega$'' (the conditions are opposite).
Formally, a \emph{game} is given by $\G = (\A,\Omega)$ where $\A$ is an arena and $\Omega$ a condition.
A condition $\Omega$ is prefix-independent if it is closed under adding and removing prefixes.
Given an arena $\A$, a subset $U$ of vertices induces a subarena
if all vertices in $U$ have an outgoing edge to $U$.
We denote by $\A[U]$ the induced arena.

\medskip\noindent{\bf Strategies.}
A \emph{strategy} for a player is a function that prescribes, given a finite history of the play, the next move.
Formally, a \emph{strategy} for Eve is a function $\sigma : V^* \cdot \VE \to V$ 
such that for a finite history $w \in V^*$ and a current vertex $v \in \VE$, the prescribed move is legal,
\textit{i.e} along an edge: $(v,\sigma(w \cdot v)) \in E$.
Strategies for Adam are defined similarly, and usually denoted by~$\tau$.
Once a game $\G = (\A,\Omega)$, a starting vertex $v_0$ and strategies $\sigma$ for Eve and $\tau$ for Adam are fixed, 
there is a unique play denoted by $\pi(v_0,\sigma,\tau)$, 
which is said to be winning for Eve if it belongs to $\Omega$.
The sentence ``Eve has a winning strategy from $U$'' means that she has a strategy
such that for all initial vertex $v_0$ in $U$,
for all strategies $\tau$ for Adam, the play $\pi(v_0,\sigma,\tau)$ is winning.
By ``solving the game'', we mean (algorithmically) determine the winner.
We denote by $\WE(\G)$ the set of vertices from where Eve wins, also referred as winning set, or winning region, 
and analogously $\WA(\G)$ for Adam.
Whenever the arena $\A$ is clear from the context, we use $\WE(\Omega)$ instead of $\WE(\A,\Omega)$.
A very important theorem in game theory, due to Martin~\cite{Martin75}, states that Borel games (that is,
where the condition is Borel, a topological condition) are determined, \textit{i.e}
we have $\WE(\G) \cup \WA(\G) = V$: from any vertex, exactly one of the two players has a winning strategy.
Throughout this paper, we only consider Borel conditions, hence our games are determined.

\medskip\noindent{\bf Memory structures.}
We define memory structures and strategies relying on memory structures.
A \emph{memory structure} $\M = (M, m_0,\up)$ for an arena $\A$ and an initial vertex $v_0$
consists of a set $M$ of memory states, 
an initial memory state $m_0 \in M$, and an update function $\up: M \times E \to M$. 
A memory structure is similar to an automaton synchronized with the arena: 
it starts from $m_0$ and reads the sequence of edges produced by the arena.
Whenever an edge is taken, the current memory state is updated using the update function $\up$.
A strategy relying on a memory structure $\M$, whenever it picks the next move, 
considers only the current vertex and the current memory state: 
it is thus given by a next-move function $\nex: \VE \times M \to V$.
Formally, given a memory structure $\M$ and a next-move function $\nex$, 
we can define a strategy $\sigma$ for Eve by $\sigma(w \cdot v) = \nex(v, \up^+(w \cdot v))$,
where $\up$ is extended to $\up^+ : V^+ \to M$.
A strategy with memory structure $\M$ has finite memory if $M$ is a finite set.
It is \emph{memoryless}, or \emph{positional} if $M$ is a singleton: in this case, the choice for the next move
only depends on the current vertex,
and can be described as a function $\sigma: \VE \to V$.

We can make the synchronized product explicit: an arena $\A$ and a memory structure $\M$ for $\A$ induce 
the expanded arena $\A \times \M = (V \times M, (\VE \times M, \VA \times M), E \times \up)$ 
where $E \times \up$ is defined by
$((v,m), (v',m')) \in E \times \up$ if $(v,v') \in E$ and $\up(m,(v,v')) = m'$.
There is a natural one-to-one mapping between plays in $\A$ and in $\A \times \M$,
and also from memoryless strategies in $\A \times \M$
to strategies in $\A$ using $\M$ as memory structure.
It follows that if a player has a memoryless strategy for the arena $\A \times \M$, then
he has a strategy using $\M$ as memory structure for the arena $\A$,
producing the same plays.
This \textit{key} property will be used throughout the paper.

\medskip\noindent{\bf Attractors.}
Given $F \subseteq V$, define $\Pre(F)$ as the union of
$\set{u \in \VE \mid \exists (u,v) \in E, v \in F}$
and $\set{u \in \VA \mid \forall (u,v) \in E, v \in F}$.
The attractor sequence is the step-by-step computation of the least fixpoint of the monotone function
$X \mapsto F \cup \Pre(X)$:
$$\left \{
\begin{array}{l}
\Attreve_0 (F) = F \\[1.4ex]
\Attreve_{k+1} (F) = \Attreve_k(F) \cup \Pre(\Attreve_k(F))
\end{array}
\right.$$
The sequence $(\Attreve_k (F))_{k \ge 0}$ is increasing with respect to set inclusion,
so it has a limit\footnote{Here we use the assumption that the set of vertices is countable.
We could drop this assumption and define the sequence indexed by ordinals, which we avoided for the sake
of readability.}, denoted $\Attrevef$, the attractor to $F$.
An attractor strategy to $F \subseteq V$ for Eve is a memoryless strategy 
that ensures from $\Attreve(F)$ to reach $F$ within a finite number of steps.
Specifically, an attractor strategy to $F$ from $\Attreve_N(F)$ 
ensures to reach $F$ within the next $N$ steps.

\medskip\noindent{\bf $\omega$-regular conditions.} 
We define the B\"uchi and parity conditions.
We equip the arena with a coloring function $c : V \rightarrow [d]$ where
$[d] = \set{0,\ldots,d}$ is the set of \emph{colors} or \emph{priorities}.
For a play $\play$, let $\infi(\play) \subseteq [d]$ 
be the set of colors that appear infinitely often in~$\play$.
The parity condition is defined by $\parp = \set{\play \mid \min(\infi(\play)) \mbox{ is even}}$,
\textit{i.e} it is satisfied if the lowest color visited infinitely often is even.
Here, the color set $[d]$ is interpreted as a set of priorities,
even priorities being ``good'' and odd priorities ``bad'',
and lower priorities preferable to higher ones.
The parity conditions are self-dual, meaning that the completement
of a parity condition is another parity condition:
$\co\parp = \Pi \setminus \parp = \parity(c+1)$.
As a special case, the class of B\"uchi conditions are defined 
using the color set $[1] = \set{0,1}$ (\textit{i.e} $d = 1$). 
We define the B\"uchi set $F$ as $c^{-1}(0) \subseteq V$,
say that a vertex is B\"uchi if it belongs to $F$,
and define $\bucf  =  \set{\pi \mid 0 \in \infi(\play)}$, \textit{i.e} 
the B\"uchi condition $\bucf$ requires that infinitely many times vertices in $F$ are reached.

The dual is $\co\bucf$ condition, which requires that finitely many times vertices in $F$ are reached.

\medskip\noindent{\bf $\omega B$-conditions.} 
We equip the arena with $k$ counters and an update function $C : E \rightarrow \set{\varepsilon,i,r}^k$,
associating to every edge an action for each counter.
The value of a counter along a play is incremented by the action $i$, reset by $r$ and left unchanged by $\varepsilon$.
We say that a counter is bounded along a play if the set of values assumed is finite,
and denote by $\bounded$ the set of plays where all counters are bounded,
and $\bounded(N)$ if bounded by $N$.
The conditions of the form $\bounded \cap \parp$ are called $\omega B$-conditions.

Note that the bound requirement for $\omega B$-conditions is not uniform: a strategy is winning if for all plays,
there exists a bound $N$ such that the counters are bounded by $N$ and the parity condition is satisfied.
In other words, the bound $N$ depends on the path.
The sentence ``Eve wins for the bound $N$'' means that Eve has a strategy which ensures
the bound $N$ uniformly: for all plays, the counters are bounded by the same $N$.
Similarly, the sentence ``the strategy (for Adam) fools the bound $N$'' means that 
it ensures that for all plays, either some counter reaches the value $N$
or the parity condition is not satisfied.

\medskip\noindent{\bf Finitary conditions.} 
Finitary conditions add bounds requirements 
over $\omega$-regular conditions~\cite{AlurHenzinger98}.
Given a coloring function $c : V \rightarrow [d]$, and a position $k$ we define:
\[
\distk (\play,c) =
\inf_{k' \ge k} \left\{
k'-k \mid 
\begin{array}{c}
c(\play_{k'}) \mbox{ is even, and } \\
c(\play_{k'}) \leq c(\play_k)
\end{array}
\right\};
\]
\textit{i.e} $\distk(\play,c)$ is the ``waiting time'' by means of number of steps 
from the $k$\textsuperscript{th} vertex to a preferable priority
(that is, even and lower).
The finitary parity winning condition $\fin\parp$ 
was defined as follows in~\cite{ChatterjeeHenzingerHorn09}:
$\fin\parp = \set{\play \mid \limsup_k \distk (\play,c) < \infty}$,
\textit{i.e} the finitary parity condition requires that the supremum 
limit of the distance sequence is bounded.
A good intuition is to see the finitary parity condition as bounding the waiting time
between requests, which are odd priorities, and responses, which are even priorities.
In this terminology, the priority $3$ is a request, answered by $0$ and $2$ since they are smaller, 
but not by $4$.
The finitary parity condition is satisfied by a play if there exists $N \in \N$
such that from some point onwards, all requests are answered within $N$ steps.

In the special case where $d = 1$, this defines the finitary B\"uchi condition:
setting $F = c^{-1} (0)$, we denote $\distk (\play,F) = \inf \set{k'-k \mid k' \geq k, \play_{k'} \in F}$,
\textit{i.e} $\distk(\play,F)$ is the number of steps from the $k$\textsuperscript{th} vertex 
to the next B\"uchi vertex.
(Note that this is consistent with the previous notation $\distk(\play,c)$.)
Then $\fin\bucf = \set{\play \mid \limsup_k \distk (\play,F) < \infty}$.
In the context of finitary conditions, the sentence 
``the strategy (for Adam) fools the bound $N$'' means that the strategy
ensures that for all plays, there exists a position $k$ such that $\distk(\pi,c) > N$.

We shall refer to games with $\omega B$-conditions as $\omega B$ games, and the same
applies for all kinds of conditions.

\begin{remark}
As defined, finitary conditions do not form a subclass of $\omega B$-conditions; however,
there exists a deterministic $\omega B$-automaton which recognizes
$\fin\parp$, so finitary games reduce to $\omega B$ games by composing with this deterministic automaton.
We informally describe this automaton: it has a counter for each odd priority,
and keeps track of the set of open requests.
As long as a request is open, the corresponding counter is incremented at each step,
and it is reset whenever the request is answered.
\end{remark}

\begin{example}\label{ex:intro}
We conclude this section by an example witnessing the difference between playing a B\"uchi condition
and a finitary B\"uchi condition over an infinite graph.
This is in contrast to the case of finite graphs, where winning for B\"uchi and finitary B\"uchi conditions are equivalent.
Figure~\ref{fig:difference_buchi_finbuchi} presents an infinite graph where only Adam has moves;
he loses the B\"uchi game but wins for the finitary B\"uchi game.
We give two representations: on the left as a pushdown graph (defined in Section~\ref{sec:pushdown}),
and on the right explicitely as an infinite-state graph.
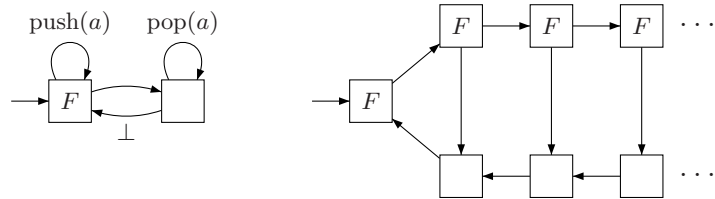
\begin{figure}
\begin{center}
\begin{picture}(80,25)(0,0)
	\rpnode[Nmarks=i,polyangle=45](push)(0,10)(4,4){$F$}
	\rpnode[polyangle=45](pop)(15,10)(4,4){}

	\drawedge[curvedepth=2](push,pop){}
	\drawedge[curvedepth=2](pop,push){$\bot$}
	\drawloop[loopdiam=5,loopangle=90](push){$\push(a)$}
	\drawloop[loopdiam=5,loopangle=90](pop){$\pop(a)$}

	\rpnode[Nmarks=i,polyangle=45](v_0)(40,10)(4,4){$F$}

	\rpnode[polyangle=45](v_1)(52,20)(4,4){$F$}
	\rpnode[polyangle=45](v_2)(64,20)(4,4){$F$}
	\rpnode[polyangle=45](v_3)(76,20)(4,4){$F$}
	\put(81,20){\begin{large}$\ldots$\end{large}}

	\rpnode[polyangle=45](u_1)(52,0)(4,4){}
	\rpnode[polyangle=45](u_2)(64,0)(4,4){}
	\rpnode[polyangle=45](u_3)(76,0)(4,4){}
	\put(81,0){\begin{large}$\ldots$\end{large}}

	\drawedge(v_0,v_1){}
	\drawedge(u_1,v_0){}
	\drawedge(v_1,v_2){}
	\drawedge(u_2,u_1){}
	\drawedge(v_1,u_1){}
	\drawedge(v_2,u_2){}
	\drawedge(v_2,v_3){}
	\drawedge(u_3,u_2){}
	\drawedge(v_3,u_3){}
\end{picture}
\end{center}
\caption{Adam loses the B\"uchi game but wins the finitary B\"uchi game.}
\label{fig:difference_buchi_finbuchi}
\end{figure}

A play consists in rounds, each starting whenever the pebble hits the leftmost vertex.
In a round, Adam chooses a number $N$ and follows the top path for $N$ steps, remaining in B\"uchi vertices; 
then he goes down, and follows a path of length $N$ without B\"uchi vertices,
before getting back to the leftmost vertex.
Whatever Adam does, infinitely many B\"uchi vertices will be visited,
so Adam loses the B\"uchi game.
However, by describing an unbounded sequence
(\textit{e.g} for $N$ steps in the $N$\textsuperscript{th} round),
Adam ensures longer and longer paths without B\"uchi vertices,
hence wins the finitary B\"uchi game.
\end{example}

%% file: strategy_complexity.tex
\subsection{Bounded and uniform conditions}

To obtain Theorem~\ref{thm:strategy_complexity},
we take five steps, summarized in Figure~\ref{fig:results},
which involve two variants of finitary conditions: uniform and bounded.
\begin{figure}[!ht]
\begin{center}
\begin{picture}(80,10)(0,0)
	\gasset{Nadjust=wh,Nframe=n,AHLength=5,AHlength=3}

	\node(bndunibuchi)(0,5){\begin{tabular}{ccc}bounded\\uniform\\B\"uchi\end{tabular}}
	\node(unibuchi)(20,5){\begin{tabular}{ccc}uniform\\B\"uchi\end{tabular}}
	\node(finbuchi)(40,5){\begin{tabular}{ccc}finitary\\B\"uchi\end{tabular}}
	\node(bndparity)(60,5){\begin{tabular}{ccc}bounded\\parity\end{tabular}}
	\node(finparity)(80,5){\begin{tabular}{ccc}finitary\\parity\end{tabular}}

  	\drawedge(bndunibuchi,unibuchi){}
  	\drawedge(unibuchi,finbuchi){}
  	\drawedge(finbuchi,bndparity){}
  	\drawedge(bndparity,finparity){}
\end{picture}
\caption{Results implications.}
\label{fig:results}
\end{center}
\end{figure}
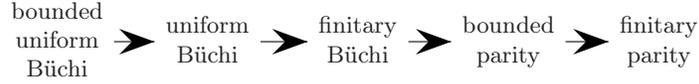

\medskip\noindent{\bf Uniform conditions.} 
Unlike finitary conditions, the bound $N \in \N$ is made explicit;
for instance the uniform B\"uchi condition is 
$\bucfn = \set{\play \mid \limsup_k \distk (\play,F) \leq N}$.

\medskip\noindent{\bf Bounded conditions.} 
Unlike finitary conditions, the requirement is not in the limit, but from the start of the play,
\textit{i.e} the distance function is bounded rather than eventually bounded;
for instance the bounded parity condition is 
$\bnd\parp = \set{\play \mid \sup_k \distk (\play,c) < \infty}$.

The two variants can be combined, for instance the bounded uniform B\"uchi condition is
defined as $\bnd\bucfn = \set{\play \mid \sup_k \distk (\play,F) \leq N}$.
Let us point out that in the special case of B\"uchi conditions, we have 
$\bnd\bucf = \fin\bucf$,
hence we can refer to these conditions either as bounded B\"uchi or as finitary B\"uchi.

\subsection{Constructing positional strategies}

We start with two general techniques to construct positional strategies.
Both techniques are about composing several positional strategies into one.
The first lemma deals with union.

\begin{lemma}[Union and positional strategies~\cite{Gimbert04}]
\label{lem:union}
Let $\A$ be an arena and $(\Omega_n)_{n \in \N}$ be a family of Borel conditions. 
If $\cup_{n \in \N} \Omega_n$ is prefix-independent and for all $n \in \N$, 
Eve has a positional winning strategy for the condition $\Omega_n$ from $V_n$, 
then she has a positional winning strategy for the condition $\cup_{n \in \N} \Omega_n$ from $\cup_{n \in \N} V_n$.
\end{lemma}

\begin{proof}
We denote by $\Omega$ the condition $\cup_{n \in \N} \Omega_n$.

For all $n \in \N$, let $\sigma_n$ be a positional strategy winning from $V_n$ 
for the condition $\Omega_n$.
We construct $\sigma$ positional strategy on $\cup_{n \in \N} V_n$: 
for $v \in \cup_{n \in \N} V_n$,
we define $\sigma(v) = \sigma_k(v)$ 
where $k$ is the smallest integer such that $v \in V_k$.
Consider a play $\pi$ consistent with $\sigma$ from $\cup_{n \in \N} V_n$:
it can be decomposed into finitely many infixes, each consistent with some strategy $\sigma_k$.
Furthermore, the index $k$ decreases along the play, hence is ultimately constant,
so $\pi$ is ultimately consistent with some $\sigma_k$.
Since $\Omega$ is prefix-independent, $\pi$ is in $\Omega$,
hence $\sigma$ is a positional winning strategy from $\cup_{n \in \N} V_n$
for the condition $\Omega$.
\hfill\qed
\end{proof}

The second lemma is about fixpoint iteration.

\begin{lemma}[Fixpoint and positional strategies]
\label{lem:fixpoint}
Let $\G = (\A,\Omega)$ be a game, where $\Omega$ is Borel and prefix-independent.
If there exists an operator $\Xi$ which associates to each subarena $\A'$ of $\A$ 
a subset of vertices of $\A'$ satisfying the following properties,
for all subarenas $\A'$:
\begin{enumerate}
	\item $\Xi(\A') \subseteq \WE(\A',\Omega)$.
	\item If $\WE(\Omega)$ is non-empty then $\Xi(\A')$ is non-empty.
	\item Eve has a positional winning strategy from $\Xi(\A')$ in the game $(\A',\Omega)$.
\end{enumerate}
then Eve has a positional winning strategy for Eve from her winning set in $\G$.
\end{lemma}

This technique will be used several times in the paper 
(see \textit{e.g}~\cite{Kopczynski06} for similar fixpoint iterations).
It consists in decomposing the winning set for~$\Omega$
into a sequence of disjoint subarenas called ``slices'',
and define a positional strategy for each slice.
Aggregating all those strategies yields a positional winning strategy 
for~$\Omega$.

\begin{proof}
We define by induction the following objects:
\begin{itemize}
	\item a sequence $(\A_k)_{k \ge 0}$ of subarenas of $\A$,
	\item a sequence of slices $(S_k)_{k \ge 1}$, 
	\item a sequence of positional strategies $(\sigma_k)_{k \ge 1}$ for Eve from $\Attreve(S_k)$.
\end{itemize}

\noindent The first arena $\A_0$ is $\A$.
%the first slice $S_0$ is $\Attreve(\Xi(\A_0))$,
%and $\sigma_0$ is an attractor strategy on $\Attreve(\Xi(\A_0)) \setminus \Xi(\A_0)$
%and a positional winning strategy from $\Xi(\A_0)$ in the game $(\A_0,\Omega)$.
Having defined $\A_k$, we set $S_{k+1} = \Attreve(\Xi(\A_k))$, 
$\A_{k+1} = \A_k \setminus S_{k+1}$ and $\sigma_{k+1}$ as 
an attractor strategy on $\Attreve(\Xi(\A_k)) \setminus \Xi(\A_k)$
and a positional winning strategy from $\Xi(\A_k)$ in the game $(\A_k,\Omega)$.

First observe that the union $S$ of all slices is the winning region for Eve in $\G$,
this follows from \textit{1.} and \textit{2.}.
Denote by $\sigma$ the union of all strategies $\sigma_k$ (note that the slices are pairwise disjoint).
The second key observation is that a play consistent with $\sigma$ from $S$ can only go down the slices, 
so eventually remains in one slice, hence is eventually consistent with some $\sigma_k$,
and as a consequence is in $\Omega$.
Thus $\sigma$ is a positional winning strategy from Eve's winning region in $\G$.
\hfill\qed
\end{proof}

\subsection{Strategy complexity for bounded uniform B\"uchi games}
Our first step is the study of bounded uniform B\"uchi games.
In this subsection, we obtain the following results:

\begin{proposition}[Strategy complexity for bounded uniform B\"uchi games]
\label{prop:mem_bounded_uniform_buchi}
For all bounded uniform B\"uchi games with bound $N$, the following assertions hold:
\begin{enumerate}
	\item Eve has a positional winning strategy from her winning set.
	\item Adam has a finite-memory winning strategy with $N$ memory states from his winning set. 
	\item In general, winning strategies for Adam require at least $N-1$ memory states, 
even over finite arenas, for $N \geq 3$. 
\end{enumerate}
\end{proposition}

We start by showing that Eve's winning set can be described using a greatest fixpoint,
which allows to define a positional winning strategy.
We define the following sequence $(Z_k)_{k \ge 0}$ of subsets of $V$:
$$\left \{
\begin{array}{lll}
Z_0 = V \\
Z_{k+1} = \Attreve_N (F \cap \Pre(Z_k)) \\
\end{array}
\right.$$
This sequence is decreasing with respect to set inclusion,
so it has a limit denoted by $Z$, equivalently defined as 
the greatest fixpoint of the monotone function 
$X \mapsto \Attreve_N (F \cap \Pre(X))$.

\begin{lemma}
$$Z = \WE(\bnd\bucfn)\ .$$
\end{lemma}

\begin{proof}
We prove both inclusions.
\begin{itemize}
	\item We first show that $Z \subseteq \WE(\bnd\bucfn)$.
Let $\sigma^N$ be a positional strategy 
that ensures from $\Attreve_N (F \cap \Pre(Z))$ to reach $F \cap \Pre(Z)$ within $N$ steps.
We define a strategy $\sigma$ on $Z$ by:
$$\sigma(v) =
\begin{cases}
\sigma^N(v) & \textrm{if } v \in \Attreve_N (F \cap \Pre(Z)) \setminus F \cap \Pre(Z) \\
v' \in Z & \textrm{if } v \in F \cap \Pre(Z)
\end{cases}$$

Consider $\play = v_0 v_1 \ldots$ a play starting from $v_0 \in Z$ consistent with $\sigma$.
By definition of $\sigma^N$ it will reach $F \cap \Pre(Z)$ within $N$ steps, 
say at vertex $v_{k_0}$ for $0 \leq k_0 \leq N$.
Furthermore the play $v_{k_0 + 1} \ldots$ is consistent with $\sigma$
and starts from $v_{k_0 + 1} \in Z$,
so repeating this reasoning by induction, we show that $\play$ visits $F$ infinitely often,
and that the distance to the next B\"uchi vertex remains smaller than $N$.
Thus $\sigma$ is a positional winning strategy for $\bnd\bucfn$ from $Z$.

	\item We now show that $V \setminus Z \subseteq \WA(\bnd\bucfn)$.
Consider a vertex $v$ not in $Z$,
we define its rank to be the smallest $k$ 
such that $v$ does not belong to $Z_k$;
note that the rank cannot be $0$.
A vertex of rang $k + 1$ belongs to $Z_k$ 
but not to $Z_{k + 1}$.
For each $k$ we define a strategy $\tau_k$:
\begin{itemize}
	\item For $k \neq 0$, the strategy $\tau_k$ ensures that
from $V \setminus \Attreve_N (F \cap \Pre(Z_k))$, 
if a B\"uchi vertex $v$ is reached within $N$ steps, 
then it does not belong to $\Pre(Z_k)$.
	\item For $k = 0$, the strategy $\tau_0$ ensures that 
from $V \setminus \Attreve_N (F)$,
no B\"uchi vertex is reached within $N$ steps.
\end{itemize}

We now define a strategy $\tau$ from $Z$: 
from a vertex of rank $k + 1$, play consistently with $\tau_k$ for $N$ steps
or until a B\"uchi vertex $v$ is reached, whichever comes first.
In the first case, Adam wins, and in the second, by definition of $\tau_k$,
$v$ is not in $\Pre(Z_k)$.
Either $v$ belongs to Eve and any successor will be in $V \setminus Z_k$,
or it belongs to Adam and the strategy $\tau$ chooses a successor in $V \setminus Z_k$.
Denote $v'$ the successor of $v$: since it is in $V \setminus Z_k$, it has a smaller rank
than $v$.
From this vertex $v'$, restart from scratch.

We argue that $\tau$ is a winning strategy from $V \setminus Z$.
Indeed, consider a play $\pi = v_0 v_1 \ldots$ from $V \setminus Z$
consistent with $\tau$.
If $v_0$ has rank $k + 1$, then
either within $N$ steps no B\"uchi vertices are visited (hence Adam wins)
or its successor has a lower rank, and the play starting from this successor is consistent with $\tau$.
Since there is no infinite decreasing sequence of integers, 
the first situation occurs and the play $\pi$ does not satisfy the bounded uniform B\"uchi condition.
Hence $\tau$ is a winning strategy from $V \setminus Z$.
\end{itemize}

\hfill\qed
\end{proof}

\vskip2em
So far, we proved that in bounded uniform B\"uchi games, 
Eve has a positional winning strategy from her winning set.

It is not clear from this characterization how to implement a winning finite-memory strategy
for Adam. To prove the finite-memory determinacy for Adam, we rely on a reduction to safety games,
that we present now.
Define the memory structure $\M = (\set{0,\ldots,N}, 0, \up)$ as:
$$\up(i,(v,v')) = 
\begin{cases}
0 & \textrm{if } v \in F \textrm{ or } v' \in F\\
i+1 & \textrm{if } i < N \textrm{ and } v,v' \notin F \\
N & \textrm{otherwise} \\
\end{cases}$$
Intuitively, the memory structure counts the number of steps
since the last visit to a B\"uchi vertex.
Then $(\G,\bnd\bucfn)$ is equivalent to 
$(\G \times \M, \safety(V \times \set{0,\ldots,N-1}))$.
Since in a safety game Adam has a positional winning strategy from his winning set,
we deduce a finite-state winning strategy using $\M$ as memory structure from his winning set in $\G$. 
Moreover, a winning strategy using $\M$ does not make use of the additional memory state $N$,
hence it actually uses $N$ memory states, and not $N+1$.

\vskip1em
Note that the positional result for Eve cannot be obtained from this reduction.
The following example shows that the upper bound given above is (almost) tight.
\begin{example}
\label{ex:memory_lower_bound_adam_uniform_buchi}
Figure~\ref{fig:memory_lower_bound_adam_uniform_buchi} presents
an arena where Adam wins for the condition $\co\bnd\buc(F,N+1)$ using $N$ memory states
and loses with less.
Here $N \geq 2$.
\begin{figure}
\begin{center}
\begin{picture}(95,40)(0,0)
  	\rpnode[polyangle=45,Nmarks=i](0)(5,20)(4,4){$c$}
	\imark[iangle=200](0)
	\imark[iangle=160](0)

  	\node[linecolor=White](1)(20,40){}
  	\node[Nw=6,Nh=6](i)(20,20){$v_{\neq i}$}
  	\node[linecolor=White](b)(20,0){}

  	\drawedge(0,1){}
  	\drawedge(0,i){}
  	\drawedge(0,b){}

 	\gasset{Nw=.8,Nh=.8,Nfill=y}
 	\node(pt1)(20,12){}
 	\node(pt2)(20,8){}
 	\node(pt3)(20,4){}
 	\node(pt4)(20,28){}
 	\node(pt5)(20,32){}
 	\node(pt6)(20,36){}

 	\gasset{Nw=6,Nh=6,Nfill=n}

  	\node[Nmarks=r](11)(40,35){$1$}
  	\node(1b)(80,35){$N$}

  	\drawedge[dash={3 1 1 1}0](11,1b){}

  	\node(i1)(40,25){$1$}
  	\node[Nmarks=r](ii)(55,25){$i$}
  	\node(ib)(80,25){$B$}

  	\drawedge[dash={3 1 1 1}0](i1,ii){}
  	\drawedge[dash={3 1 1 1}0](ii,ib){}

  	\node(j1)(40,15){$1$}
  	\node[Nmarks=r](jj)(65,15){$j$}
  	\node(jb)(80,15){$N$}

  	\drawedge[dash={3 1 1 1}0](j1,jj){}
  	\drawedge[dash={3 1 1 1}0](jj,jb){}

  	\node(b1)(40,5){$1$}
  	\node[Nmarks=r](bb)(80,5){$N$}

  	\drawedge[dash={3 1 1 1}0](b1,bb){}

  	\drawedge(i,11){}
  	\drawedge(i,j1){}
  	\drawedge(i,b1){}
  	
  	\node[linecolor=White](invisible)(95,20){\large{\ to $c$}}
  	\drawedge(1b,invisible){}
  	\drawedge(ib,invisible){}
  	\drawedge(jb,invisible){}
  	\drawedge(bb,invisible){}
\end{picture}
\end{center}
\caption{An arena where Adam needs $N$ memory states to win a bounded uniform B\"uchi game.}
\label{fig:memory_lower_bound_adam_uniform_buchi}
\end{figure}
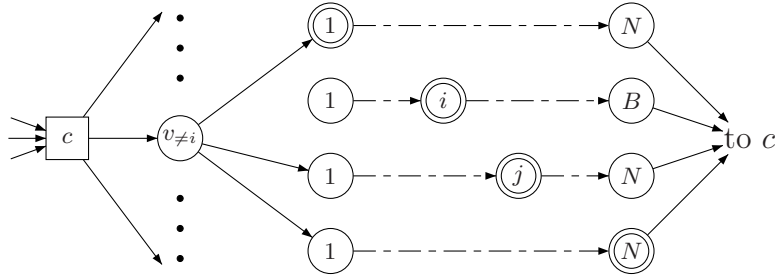
A play consists in repeating infinitely many times the following interaction:
first, from $c$ Adam chooses an $i$ from $\set{1,\ldots,N}$, 
then from $v_{\neq i}$ Eve chooses a $j$ different from $i$,
and follows a path of length $N$ where only the $j$\textsuperscript{th} vertex is B\"uchi.
Adam wins using $N$ memory states by playing the last choice of Eve: 
this way, either Eve chooses a $j$ larger than $i$ 
so no B\"uchi vertices will be visited within $N + 2$ steps, 
or she chooses a $j$ smaller than $i$. 
The first case occurs infinitely many times, 
so the uniform B\"uchi condition is violated.
If Adam uses less than $N$ memory states, 
then there exists an $i$ that he will never choose:
Eve wins $\bnd\buc(F,N+1)$ by choosing $i$ every time.
\end{example}

\subsection{Strategy complexity for uniform B\"uchi games}
\label{subsec:strategy_complexity_uniform_buchi}
Our second step is about uniform B\"uchi games.
In this subsection, we obtain the following results:

\begin{proposition}[Strategy complexity for uniform B\"uchi games]
\label{prop:mem_uniform_buchi}
For all uniform B\"uchi games with bound $N$, the following assertions hold:
\begin{enumerate}
	\item Eve has a positional winning strategy from her winning set.
	\item Adam has a finite-memory winning strategy with $N+1$ memory states from his winning set. 
	\item In general, winning strategies for Adam require at least $N-1$ memory states, 
even over finite arenas, for $N \geq 2$.
\end{enumerate}
\end{proposition}

The bounded uniform B\"uchi conditions are the prefix-dependent counterpart 
of the uniform B\"uchi conditions: 
$$\bucfn = V^* \cdot \bnd\bucfn \ .$$
However, this does not imply the equality between 
$\WE(\bucfn)$ and $\Attreve(\WE(\bnd\bucfn))$.
One inclusion holds: $$\Attreve(\WE(\bnd\bucfn)) \subseteq \WE(\bucfn)\ ,$$
but the other fails, as shown in Figure~\ref{fig:bounded_uniform_vs_uniform}.

\begin{figure}
\begin{center}
\begin{picture}(40,10)(0,0)
	\rpnode[Nmarks=i,polyangle=45](v_0)(0,2)(4,4){$F$}
	\rpnode[polyangle=45](v_1)(15,2)(4,4){}
	\rpnode[polyangle=45](v_2)(30,2)(4,4){$F$}

	\drawloop[loopdiam=4,loopangle=90](v_0){}
	\drawedge(v_0,v_1){}
	\drawedge(v_1,v_2){}
	\drawloop[loopdiam=4,loopangle=90](v_2){}
\end{picture}
\end{center}
\caption{$\Attreve(\WE(\bnd\buc(F,0))) \subsetneq \WE(\buc(F,0))$}
\label{fig:bounded_uniform_vs_uniform}
\end{figure}

This shows that one iteration of the bounded uniform B\"uchi winning set does not
give the whole uniform B\"uchi winning set.
However, the following properties hold:
\begin{enumerate}
	\item $\WE(\bnd\bucfn) \subseteq \WE(\bucfn)$,
	\item if $\WE(\bucfn)$ is non-empty then $\WE(\bnd\bucfn)$ is non-empty.
\end{enumerate}
The first item is clear, we prove the second. 
Assume $\WE(\bnd\bucfn) = \emptyset$,
then $\WA(\bnd\bucfn) = V$: from everywhere Adam can fool the bound $N$.
Iterating such strategies, he can fool the bound $N$ infinitely often,
so $\WA(\bucfn) = V$, which implies $\WE(\bucfn) = \emptyset$.

\noindent We apply Lemma~\ref{lem:fixpoint} with the operator $\Xi$ 
that associates to each subarena $\A'$ the set $\WE(\A',\bnd\bucfn)$.
The first two properties \textit{1.} and \textit{2.} have been proved above,
and the third one is a consequence of Proposition~\ref{prop:mem_bounded_uniform_buchi},
since $V^* \cdot \bnd\bucfn \subseteq \bucfn$.
It follows that in uniform B\"uchi games, Eve has a positional winning strategy from her winning set.

\vskip1em
The proof of the results for Adam follows the same lines as above.
We first lift up the reduction, which is now from uniform B\"uchi games to CoB\"uchi games.
The memory structure is the same as above, and now $(\G,\bucfn)$ is equivalent to 
$(\G \times \M, \co\buc(V \times \set{0,\ldots,N-1}))$.
Since in a CoB\"uchi game, Adam has a positional winning strategy from his winning set, 
we deduce a finite-state winning strategy using $\M$ as memory structure from his winning set in $\G$. 
Notice that this gives an upper bound of $N+1$ memory states, whereas in the case of
bounded uniform B\"uchi games, we had an upper bound of $N$ memory states.

We now discuss the lower bound: we can easily see that the statements 
about the game presented in Example~\ref{ex:memory_lower_bound_adam_uniform_buchi}
hold true for bounded uniform B\"uchi conditions as well as for uniform B\"uchi conditions,
hence the same lower bound of $N-1$ applies.

\subsection{Strategy complexity for finitary B\"uchi games}
Our third step is about finitary B\"uchi games.
In this subsection, we obtain the following results:

\begin{proposition}[Strategy complexity for finitary B\"uchi games]
\label{prop:mem_fin_buchi}
For all finitary B\"uchi games, the following assertions hold:
\begin{enumerate}
	\item Eve has a positional winning strategy from her winning set.
	\item In general winning strategies for Adam require infinite memory, 
even for pushdown arenas.
\end{enumerate}
\end{proposition}

Let $\G = (\A,\fin\bucf)$ be a finitary B\"uchi game.
We denote by $\Xi$ the operator that associates to a subarena $\A'$ 
the set of vertices $\bigcup_N \WE(\A',\bucfn)$.
To apply Lemma~\ref{lem:fixpoint}, we prove the following properties,
for all subarenas $\A'$:
\begin{enumerate}
	\item $\Xi[\A'] \subseteq \WE(\A',\fin\bucf)$.
	\item If $\WE(\fin\bucf)$ is non-empty then $\Xi[\A']$ is non-empty.
	\item Eve has a positional winning strategy from $\Xi[\A']$ in $(\A',\fin\bucf)$.
\end{enumerate}

The first item is clear, with the following interpretation in mind:
$\Xi[\A']$ is the set of vertices where Eve can announce a bound $N$ upfront and claim 
``I will win for the condition $\bucfn$''.
However, it may be that even if Eve wins, 
she is never able to announce a bound: such a situation happens 
in Example~\ref{ex:boundunknown}.

\begin{example}
\label{ex:boundunknown}
Figure~\ref{fig:boundunknown} presents an infinite one-player arena,
where Eve wins yet is not able to announce a bound.
A loop labeled $n$ denotes a loop of length $n$,
where a B\"uchi vertex is visited every $n$ steps.
In this game, as long as Adam decides to remain in the top path,
Eve cannot claim that she will win for some uniform B\"uchi condition.
\begin{figure}
\begin{center}
\begin{picture}(65,24)(0,0)
	\rpnode[Nmarks=i,polyangle=45](v_0)(0,14)(4,4){$F$}

	\gasset{Nadjust=w}

	\rpnode[polyangle=45](v_1)(20,24)(4,4){$F$}
	\rpnode[polyangle=45](v_2)(40,24)(4,4){$F$}
	\put(45,25){\begin{Huge}$\ldots$\end{Huge}}
	\rpnode[polyangle=45](v_n)(60,24)(4,4){$F$}
	\put(65,25){\begin{Huge}$\ldots$\end{Huge}}

	\rpnode[polyangle=45](u_1)(20,9)(4,4){$F$}

	\rpnode[polyangle=45](u_2)(40,12)(4,4){$F$}
	\rpnode[polyangle=45](u_2b)(40,0)(4,4){}

	\rpnode[polyangle=45](u_n)(60,15)(4,4){$F$}
	\put(65,15){\begin{Huge}$\ldots$\end{Huge}}

	\drawedge(v_0,v_1){}
	\drawedge(v_1,v_2){}
	\drawedge(v_1,u_1){}
	\drawedge(v_2,u_2){}
	\drawedge[curvedepth=5](u_2,u_2b){}
	\drawedge[curvedepth=5](u_2b,u_2){}
	\drawedge(v_n,u_n){}

	\drawloop[loopdiam=5,loopangle=-90](u_1){}
	\drawloop[loopdiam=10,loopangle=-90,dash={3 1.5}{1.5}](u_n){$n$}
\end{picture}
\end{center}
\caption{An infinite arena where Eve cannot predict the bound.}
\label{fig:boundunknown}
\end{figure}
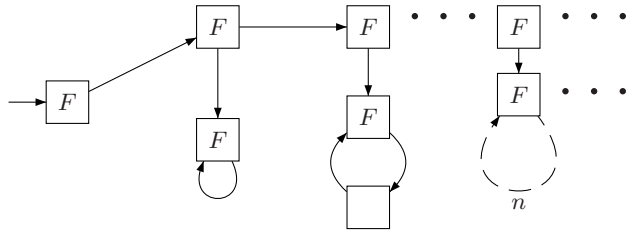
\end{example}

We prove the second item, by the contrapositive.
Assume that for all $N$, the winning set $\WE(\bucfn)$ is empty,
so Adam wins for the condition $\co\bucfn$ from everywhere:
let $\tau_N$ be a winning strategy for Adam.
From any vertex, the strategy $\tau_N$ fools the bound $N$,
\textit{i.e} for all plays consistent with $\tau_N$,
there is a sequence of $N$ consecutive non-B\"uchi vertices.
Playing in turns $\tau_1$ until such a sequence occurs, 
then $\tau_2$, and so on, ensures to spoil the condition $\fin\bucf$.
Hence Adam wins everywhere for the condition $\co\fin\bucf$,
which implies $\WE(\fin\bucf) = \emptyset$.

We now prove the third item. 
We know from Proposition~\ref{prop:mem_uniform_buchi} that
Eve has a positional winning strategy from $\WE(\bucfn)$ for the condition $\bucfn$.
Now thanks to Lemma~\ref{lem:union} we deduce that she has a positional
winning strategy from $\bigcup_N \WE(\bucfn)$ for the condition $\bigcup_N \bucfn$ (that is, $\fin\bucf$).

\vskip1em
\noindent 
By Lemma~\ref{lem:fixpoint}, in finitary B\"uchi games, 
Eve has a positional winning strategy from her winning set,
and the winning region for finitary B\"uchi is obtained
as the least fixpoint of the operator $\Xi$.

An arena where Adam needs infinite memory to win in a finitary B\"uchi game was already presented and discussed in 
Figure~\ref{fig:difference_buchi_finbuchi}.

\vskip1em
We summarize in the following theorem
the winning sets characterizations obtained 
for the three variants of B\"uchi conditions,
using mu-calculus formulae with infinite disjunction.

\begin{theorem}[Characterizations of the winning sets]
$$\WE(\bnd\bucfn) = \nu Z \cdot \Attreve_N (F \cap \Pre(Z))\ ,$$
$$\WE(\bucfn) = \mu Y \cdot \nu Z \cdot \Attreve_N ((F \cup Y) \cap \Pre(Z))\ ,$$
$$\WE(\fin\bucf) = \mu X \cdot \left(\bigcup_{N \in \N} \mu Y \cdot \nu Z \cdot \Attreve_N((F \cup Y \cup X) \cap \Pre(Z))\right)\ .$$
\end{theorem}

\subsection{Strategy complexity for bounded parity games}
Our fourth step is about bounded parity games.
In this subsection, we obtain the following results:

\begin{proposition}[Strategy complexity for bounded parity games]
\label{prop:mem_bounded_parity}
For all bounded parity games, the following assertions hold:
\begin{enumerate}
	\item Eve has a finite-memory winning strategy that uses $\dhalf$ memory states from her 
	winning set, where $\ell$ is the number of odd colors.
	\item In general, winning strategies for Eve from her winning set require two memory states
	(\textit{i.e}, positional strategies do not suffice for winning).
\end{enumerate}
\end{proposition}

%Note that the upper bound and the lower bound presented in this proposition do not match.

We present a reduction from bounded parity games to bounded B\"uchi games.
Let $\G = (\A,\bnd\parp)$ be a bounded parity game equipped 
with the coloring function $c : V \rightarrow [d]$, and assume that $d$ is even.
Define the memory structure 
$\M = (\set{1,3,\ldots,d-1} \cup \set{d}, m_0, \up)$,
where:
$$\up(m,(v,v')) = 
\begin{cases}
m & \textrm{if } c(v') \geq m \\
c(v') & \textrm{if } c(v') < m \textrm{ and } c(v') \textrm{ is odd} \\
d & \textrm{if } c(v') < m \textrm{ and } c(v') \textrm{ is even} \\
\end{cases}$$
$$m_0 = 
\begin{cases}
c(v_0) & \textrm{if } c(v_0) \textrm{ is odd}\\
d & \textrm{otherwise} \\
\end{cases}$$
Intuitively, this memory structure keeps track of 
the most important pending request.
It will be used several times in the paper, in Section~\ref{sec:pushdown} as well as in Section~\ref{sec:stack_boundedness}.

Let $F = \set{(v,d) \mid c(v) \textrm{ is even}}$.
which intuitively corresponds to the case where all requests got answered.
We argue that $\G$ is equivalent to $\G \times \M = (\A \times \M, \bnd\bucf)$,
\textit{i.e} the following are equivalent:
$$\pi \in \bnd\parp \ \ \textrm{ if and only if } \ \ \widetilde{\pi} \in \bnd\bucf,$$
where $\widetilde{\pi}$ is the play in $\G \times \M$ corresponding to $\pi$.

We prove the left-to-right direction. 
Let $\pi \in \bnd\parp$, then there exists $N$ such that for all $k$, $\distk(\pi,c) \le N$;
in other words every request is answered within $N$ steps.
We argue that in $\widetilde{\pi}$, for all positions $k$, we have $\distk(\widetilde{\pi},F) \le N \cdot \ell$.
Indeed, consider the memory states assumed along $\widetilde{\pi}$,
\textit{i.e} the set of open requests.
Since each request is answered within $N$ steps, they are removed from the memory state; 
however, it may be that along the way other requests are opened.
If they are not answered within these $N$ steps, then they are smaller.
The set of open requests can only decrease $\ell$ times, implying our claim.

Conversely, let $\widetilde{\pi} \in \bnd\bucf$, then there exists $N$ 
such that for all $k$, $\distk(\widetilde{\pi},F) \le N$.
in other words after $N$ steps no request is pending.
A fortiori, every request is answered within $N$ steps,
so $\pi \in \bnd\parp$.
This concludes.

\vskip1em
\noindent Thanks to Proposition~\ref{prop:mem_fin_buchi}, 
in a bounded B\"uchi game Eve has a positional winning strategy from her winning set,
which implies that she has a positional winning strategy using $\M$ as memory structure from her winning set
in $\G$.

Note that this does not give a reduction from finitary parity games to finitary B\"uchi games:
the above equivalence does not hold for the prefix-independent conditions.
For instance, $\pi = 1 \cdot 2^\omega$ satisfies the finitary parity condition
but $\widetilde{\pi} = (1,1) \cdot (2,1)^\omega$ does not satisfy the finitary B\"uchi condition
(the memory state remains equal to $1$ forever).

\vskip1em
We now consider the lower bounds on memory.
The fact the Eve needs memory is illustrated in Example~\ref{ex:memory_needed_bounded_parity}.
Note that from the special case of bounded B\"uchi conditions we already know an infinite lower bound for Adam.

\begin{example}\label{ex:memory_needed_bounded_parity}
Figure~\ref{fig:memory_needed_bounded_parity} presents an infinite arena,
where for condition $\bnd\parp$, 
Eve needs two memory states to win.
This is in contrast with finite arenas, 
where she has positional winning strategies~\cite{ChatterjeeHenzingerHorn09}.
The label $n$ on an edge indicates that the length of the path is $n$.
A play is divided in rounds, and a round is as follows:
first Adam makes a request, either $1$ or $3$,
and then Eve either answers both requests and proceeds to the next round,
or stops the play visiting color $2$.
Assume Eve uses a positional strategy, and consider two cases:
either she chooses always $0$, then Adam wins by choosing always $3$,
ensuring that the response time grows unbounded,
or at some round she chooses $2$, 
then Adam wins by choosing $1$ at this particular round,
ensuring that this last request will never be responded.
However, if Eve answers correctly -- that is choosing color $0$ for the request $1$,
and color $2$ for the request $3$ -- the bounded parity condition is satisfied,
and this requires two memory states.
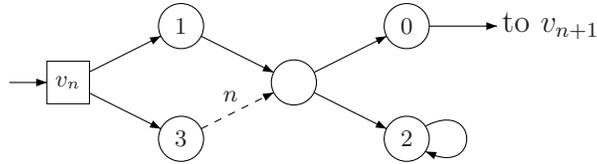
\begin{figure}
\begin{center}
\begin{picture}(60,15)(0,0)
	\rpnode[Nmarks=i,polyangle=45](v_n)(0,7.5)(4,4){$v_n$}

	\gasset{Nh=6,Nw=6}

	\node(v_1)(15,15){$1$}
	\node(v_3)(15,0){$3$}
	\node(v_c)(30,7.5){}
	\node(v_0)(45,15){$0$}
	\node(v_2)(45,0){$2$}
	
	\drawedge(v_n,v_1){}
	\drawedge(v_n,v_3){}
	\drawedge(v_1,v_c){}
	\drawedge[dash={1}0](v_3,v_c){$n$}
	\drawedge(v_c,v_0){}
	\drawedge(v_c,v_2){}

	\drawloop[loopdiam=5,loopangle=0](v_2){}

  	\node[linecolor=White](invisible)(60,15){\large{\qquad to $v_{n+1}$}}
  	\drawedge(v_0,invisible){}
\end{picture}
\end{center}
\caption{An infinite arena where Eve needs memory to win $\bnd\parp$.}
\label{fig:memory_needed_bounded_parity}
\end{figure}
\end{example}

Before proceeding to the fifth and last step, 
let us discuss why the fourth step was about bounded parity conditions
rather than uniform ones.
In both uniform parity games and bounded parity games, Eve needs memory to win;
this is shown in Example~\ref{ex:memory_needed_bounded_parity} for bounded parity conditions, 
and in Example~\ref{ex:memory_needed_uniform_parity} for uniform parity conditions.
It follows that using any of the two routes would not give positional winning strategies
for our final goal, finitary parity conditions.
Furthermore, extending the techniques for bounded B\"uchi games to bounded parity games
is quite technical, as characterizing the winning regions requires nesting least and greatest fixpoints,
whereas the reduction we described from bounded parity games to bounded B\"uchi games
is both conceptually simple and effective.

\begin{example}
\label{ex:memory_needed_uniform_parity}
Figure~\ref{fig:memory_needed_uniform_parity} presents a finite arena,
where for condition $\parity(p,2)$, 
Eve needs two memory states to win.
First Adam makes a request, either $1$ or $3$,
and then Eve chooses between $0$ and $2$.
If Eve answers correctly -- that is choosing color $0$ for the request $3$,
and color $2$ for the request $1$-- the bound requirement is satisfied,
and this requires two memory states.
Otherwise, either the bound requirement is too large 
(if she chooses color $0$ while Adam chose color $3$)
or the answer is not appropriate 
(if she chooses color $2$ while Adam chose color $1$).
This example is easily generalized to the case of $2d+1$ colors,
and there Eve needs $d+1$ memory states to answer the requests appropriately.
\begin{figure}
\begin{center}
\begin{picture}(75,25)(0,0)
	\rpnode[Nmarks=i,polyangle=45](v_0)(0,10)(4,4){$v_0$}

	\gasset{Nh=6,Nw=6}
	
	\node(v_1)(15,20){$3$}
	\node(v_2)(30,20){}
	\node(v_3)(22,0){$1$}
	\node(v_4)(40,10){}
	\node(v_5)(60,20){$2$}
	\node(v_6)(52,0){}
	\node(v_7)(65,0){$0$}
	
	\drawedge(v_0,v_1){}
	\drawedge(v_1,v_2){}
	\drawedge(v_0,v_3){}
	\drawedge(v_2,v_4){}
	\drawedge(v_3,v_4){}
	\drawedge(v_4,v_5){}
	\drawedge(v_4,v_6){}
	\drawedge(v_6,v_7){}

	\drawloop[loopdiam=4,loopangle=0](v_5){}
	\drawloop[loopdiam=4,loopangle=0](v_7){}
\end{picture}
\end{center}
\caption{An arena where Eve needs memory to win $\parity(p,2)$.}
\label{fig:memory_needed_uniform_parity}
\end{figure}
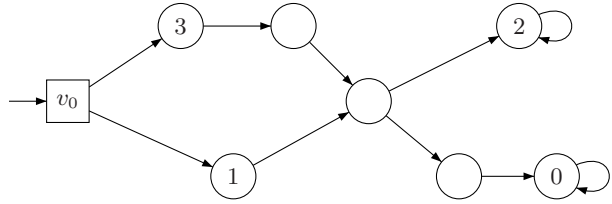
\end{example}

\subsection{Strategy complexity for finitary parity games}
Our last step is about finitary parity games.
In this subsection, we obtain the following results:

\begin{proposition}[Strategy complexity for finitary parity games]
\label{prop:mem_fin_parity}
For all finitary parity games, Eve has a finite-memory winning strategy from her 
winning set that uses at most $\dhalf$ memory states, where $\ell$ is the number of odd colors.
\end{proposition}

Once again, we rely on Lemma~\ref{lem:fixpoint} to prove this result.
Let $\G = (\A,\fin\bucf)$ be a finitary parity game,
and $\M$ the memory structure defined in the fourth step.
We consider the arena $\A \times \M$, and denote by $\Xi$ the operator that associates to a subarena $\A'$ 
of $\A \times \M$ the set of vertices $\WE(\A',\bnd\parp)$.
Specifically, we have, for all subarenas $\A'$:
\begin{enumerate}
	\item $\Xi[\A'] \subseteq \WE(\A',\fin\parp)$.
	\item If $\WE(\A',\fin\parp)$ is non-empty, then $\Xi[\A']$ is non-empty.
	\item Eve has a positional winning strategy from $\Xi[\A']$ in $(\A',\fin\parp)$.
\end{enumerate}
The proof is easy and follows the same lines as for the third step.

\vskip1em
Theorem~\ref{thm:strategy_complexity} gives the almost complete picture: the notable exception is the gap
for finitary parity games, where we prove that $\dhalf$ memory states are sufficient 
for Eve, yet without showing that any memory is required at all.
Although we think that positional strategies always exist, we were not able to prove it.
Our techniques through bounded parity games cannot be improved for this purpose,
as we showed that for these games memory is required for Eve's winning strategies.

%% file: games.tex
\medskip\noindent{\bf Pushdown arenas.}
A pushdown process is a finite-state machine which features a stack: it is described as $(Q,\Gamma,\Delta)$
where $Q$ is a finite set of control states, $\Gamma$ is the stack alphabet and $\Delta$ is the transition relation.
There is a special stack symbol denoted $\bot$ which does not belong to $\Gamma$;
we denote by $\Gamma_\bot$ the alphabet $\Gamma \cup \set{\bot}$.
A configuration is a pair $(q,u \bot)$ (the top stack symbol is the leftmost symbol of $u$).
There are three kinds of transitions in $\Delta$:
\begin{itemize}
	\item $(p,a,\push(b),q)$: allowed if the top stack element is $a \in \Gamma_\bot$, 
	the symbol $b \in \Gamma$ is pushed onto the stack.
	\item $(p,\pop(a),q)$: allowed if the top stack element is $a \in \Gamma$,
	the top stack symbol $a$ is popped from the stack.
	\item $(p,a,\Skip,q)$: allowed if the top stack element is $a \in \Gamma_\bot$, the stack remains unchanged.
\end{itemize}
The symbol $\bot$ is never pushed onto, nor popped from the stack.
The pushdown arena of a pushdown process is defined as 
$(Q \times \Gamma^* \bot, (Q_E \times \Gamma^* \bot, Q_A \times \Gamma^* \bot),E)$,
where $(Q_E,Q_A)$ is a partition of $Q$ and $E$ is given by the transition relation $\Delta$.
For instance if $(p,a,\push(b),q) \in \Delta$, 
then $\left((p,a w \bot),(q,baw \bot)\right) \in E$, for all words $w$ in $\Gamma^*$.

\medskip\noindent{\bf Conditions.}
The coloring functions for parity conditions over pushdown arenas are specified over the control states,
\textit{i.e} do not depend on the stack content.
Formally, a coloring function is given by $c : Q \rightarrow [d]$,
and extended to $c : Q \times \Gamma^* \bot \rightarrow [d]$ by $c(q,u \bot) = c(q)$.

We begin this section by giving two examples witnessing interesting phenomena of pushdown finitary games
(hence a fortiori of pushdown $\omega B$ games).

\begin{example}\label{ex:pushdown_games_example_waiting_bound}
Figure~\ref{fig:pushdown_games_example_waiting_bound} presents a pushdown finitary B\"uchi game,
where Eve wins for the bound $0$, but loses the bounded uniform condition for any bound.
Let us first look at the two bottom states: in the left-hand state at the bottom, 
Adam can push as many $b$'s as he wishes,
and moves the token to the state to its right, 
where all those $b$'s are popped one at a time.
In other words, each visit of the two bottom states allows Adam to announce a number $N$ and to prove
that he can ensure a sequence of $N$ consecutive non-B\"uchi states.
We now look at the states on the top line: the initial state is the leftmost one,
where Adam can push an arbitrary number of $a$'s.
We see those $a$'s as credits: from the central state, Adam can use one credit (\textit{i.e} pop an $a$)
to pay a visit to the two bottom states.
When he runs out of credit, which will eventually happen, he moves the token to the rightmost state,
where nothing happens anymore.
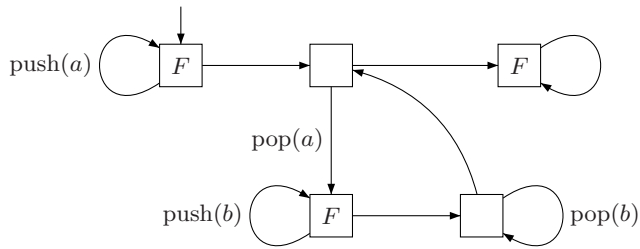
\begin{figure}[!ht]
\begin{center}
\begin{picture}(45,25)(-4,0)
	\gasset{Nh=6,Nw=6}

	\rpnode[Nmarks=i,iangle=90,polyangle=45](q_0)(0,20)(4,4){$F$}
	\rpnode[polyangle=45](q_1)(20,20)(4,4){}
	\rpnode[polyangle=45](q_2)(20,0)(4,4){$F$}
	\rpnode[polyangle=45](q_3)(40,0)(4,4){}
	\rpnode[polyangle=45](q_4)(45,20)(4,4){$F$}

	\drawloop[loopangle=180](q_0){$\push(a)$}
	\drawloop[loopangle=180](q_2){$\push(b)$}
	\drawloop[loopangle=0](q_3){$\pop(b)$}
	\drawloop[loopangle=0](q_4){}

	\drawedge(q_0,q_1){}
	\drawedge[ELside=r](q_1,q_2){$\pop(a)$}
	\drawedge(q_2,q_3){}
	\drawedge[curvedepth=-5](q_3,q_1){}
	\drawedge(q_1,q_4){}
\end{picture}
\end{center}
\caption{A pushdown game where Eve wins $\buc(F,0)$ but loses for any condition $\bnd\bucfn$.}
\label{fig:pushdown_games_example_waiting_bound}
\end{figure}
\end{example}

\begin{example}\label{ex:pushdown_games_example_switch}
Figure~\ref{fig:pushdown_games_example_switch} presents a pushdown finitary B\"uchi game where 
Eve wins for the bound $2$, but to do this she has to maintain a small stack.
A play in this game divides into infinitely many rounds, which start by a visit to $q$.
As in the previous example, each letter $a$ on the stack is a ``credit''.
A round consists in the following actions: 
first Eve chooses whether she wants to pop some $a$'s from the stack (self-loop around $q$),
and then moves the token to the B\"uchi state,
second Adam decides either to push an $a$ and start the next round or
to go to the rightmost state $p$ to pop some $a$'s.
The latter action should be understood as using credits ($a$'s on the stack) to remain away from the B\"uchi state;
using $N$ credits, he can stay in $p$ for $N$ steps.
It follows that Eve should everytime keep the stack low to avoid long stays in $p$.
This rules out the greedy (attractor) strategy for her which would rush to the B\"uchi state without considering the stack;
a wiser strategy ensuring the bound $2$ is to start every round by popping the $a$ pushed during the previous round.
\begin{figure}[!ht]
\begin{center}
\begin{picture}(40,12)(0,0)
	\gasset{Nh=6,Nw=6}

	\node[Nmarks=i,iangle=-90](q)(0,5){$q$}
	\rpnode[polyangle=45](F)(20,5)(4,4){$F$}
	\rpnode[polyangle=45](p)(40,5)(4,4){$p$}

	\drawedge[curvedepth=5](q,F){}
	\drawedge[curvedepth=5](F,q){$\push(a)$}
	\drawedge[curvedepth=5](F,p){$\pop(a)$}
	\drawedge[curvedepth=5](p,F){}

	\drawloop[loopangle=180](q){$\pop(a)$}
	\drawloop[loopangle=0](p){$\pop(a)$}
\end{picture}
\end{center}
\caption{A pushdown game with finitary B\"uchi conditions where Eve has to maintain a small stack.}
\label{fig:pushdown_games_example_switch}
\end{figure}
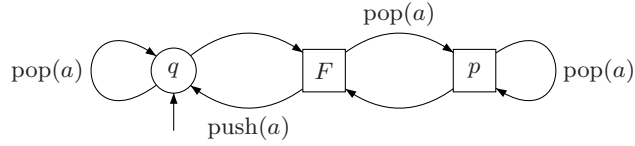
\end{example}

%% file: p_automata.tex
We will use alternating $\PP$-automata to recognize sets of configurations:
an alternating $\PP$-automaton $\B = (S,\delta,F)$ for the pushdown process $\PP = (Q,\Gamma,\Delta)$
is a classical alternating automaton over finite words:
$S$ is a finite set of control states, $\delta : S \times \Gamma \rightarrow \mathbb{B}^+(S)$
is the transition function (where $\mathbb{B}^+(S)$ is the set of positive boolean formulae over $S$),
and $F$ is a subset of $S$ of final states.
We assume that the set of states $S$ contains $Q$.
A configuration $(q,u\bot)$ is accepted by $\B$ if it is accepted
using $q \in Q \subseteq S$ as initial state, with the standard alternating semantics.
A set of configuration is said regular if it is accepted by an alternating $\PP$-automaton.

The following theorem states that for very general conditions,
the winning region is regular~\cite{Serre03,Serre06}.

\begin{theorem}[\cite{Serre06}]\label{thm:pushdown_games_regularity}
For all pushdown games, for all winning conditions $\Omega \subseteq Q^\omega$
that are Borel and prefix-independent,
the set $\WE(\Omega)$ is a regular set of configurations
recognized by an alternating $\PP$-automaton of size $|Q|$.
\end{theorem}

%% file: collapse.tex
We denote $\limitbounded(N)$ the set of plays which contain a suffix 
for which the counters are bounded by~$N$.
Note that unlike $\bounded(N)$, the condition $\limitbounded(N)$ is prefix-independent,
so Theorem~\ref{thm:pushdown_games_regularity} applies.

\begin{theorem}[The forgetful property]
\label{thm:collapse}
For all pushdown $\omega B$ games, the following are equivalent:
\begin{itemize}
	\item $\exists \sigma$ strategy for Eve, $\forall \pi$ plays, $\exists N \in \N$,\
	$\pi \in \bounded(N) \cap \parp$,
	\item $\exists \sigma$ strategy for Eve, $\exists N \in \N$, $\forall \pi$ plays, 
	$\pi \in \limitbounded(N) \cap \parp$.
\end{itemize}
\end{theorem}

We refer to this result as a collapse result, as it shows that the non-uniform quantification (with respect to bound)
of pushdown $\omega B$ games collapses to a uniform quantification (but using a slightly different bounding condition).
It follows that we can associate to a pushdown $\omega B$ game a bound $N$, called the \textit{collapse bound},
which only depends on the pushdown arena and the condition attached.
Later in this section, we will show doubly-exponential lower bounds on this collapse bound.

The intuition behind the name forgetful property is the following: 
even if a configuration carries an unbounded amount of information (since the stack may be arbitrarily large),
this information cannot be forever transmitted along a play.
Indeed, to increase the counter values significantly,
Adam has to use the stack, consuming or forgetting its original information.

Example~\ref{ex:pushdown_games_example_waiting_bound} shows that the content of the stack
can be used as ``credit'' for Adam, but also that if Eve wins
then from some point onwards this credit vanishes.
Slightly modified, it also shows that Theorem~\ref{thm:collapse} does not hold if
$\limitbounded$ is replaced by $\bounded$.

For the sake of readability, we abbreviate $\WE(\limitbounded(N) \cap \parp)$ by $\WE(N)$,
and similarly $\WE(\bounded \cap \parp)$ by $\WE$.
The following properties hold:
\begin{enumerate}
	\item $\WE(0) \subseteq \WE(1) \subseteq \WE(2) \subseteq \cdots \subseteq \WE$.
	\item There exists $N$ such that $\WE(N) = \WE(N+1) = \cdots$.
	\item For such $N$, we have $V \setminus \WE(N) \subseteq \WA$, hence $\WE = \WE(N)$.
\end{enumerate}

The first item is clear.
For the second we rely on Theorem~\ref{thm:pushdown_games_regularity}.
For every $N$ there exists $\B_N$ an alternating $\PP$-automaton of size $|Q|$
recognizing $\WE(N)$.
Since there are finitely many alternating $\PP$-automata of size $|Q|$,
the increasing sequence of the set of configurations they recognize is ultimately constant,
\textit{i.e} there exists $N$ such that $\B_N = \B_{N+1} = \ldots$.
We now argue that the third item holds.
From the complement of $\WE(N)$, Adam can ensure to fool the bound $N$, 
but also $N+1$, and so on, yet remaining there.
Iterating such strategies ensures to spoil the $\omega B$-condition,
which concludes the proof.

\begin{remark}
The above proof does not give a bound on $N$;
indeed, the sequence $(\B_N)_{N \in \N}$ is ultimately constant, 
but the fact that two consecutive automata are equal, \textit{i.e} $\B_N = \B_{N + 1}$,
does not imply that from there on the sequence is constant.
It follows that $N$ can be \emph{a priori} arbitrarily large.

We will later present examples showing that the bound $N$
is at least doubly-exponential in the number of vertices,
and exponential in the stack alphabet.
\end{remark}

%% file: decidability.tex
We give two proofs of decidability of solving pushdown games:
\begin{itemize}
	\item First, we prove the decidability of solving pushdown finitary games,
relying on the finite-memory results of Section~\ref{sec:strategy_complexity} (Theorem~\ref{thm:strategy_complexity}) 
the collapse result (Theorem~\ref{thm:collapse}), and~\cite{Bojanczyk04}.
	\item Second, we prove the decidability of solving pushdown $\omega B$ games,
generalizing the first item. 
This relies on the collapse result (Theorem~\ref{thm:collapse}) and~\cite{BlumensathColcombetKuperbergVandenboom13}.
\end{itemize}

We begin by proving the decidability of pushdown finitary games.
Note that the second property in Theorem~\ref{thm:collapse}, namely:
\[
\exists \sigma \textrm{ strategy for Eve}, \exists N \in \N, \forall \pi \textrm{ plays},\
\pi \in \limitbounded(N) \cap \parp\ ,
\]
can be written as an existential bounding formula over infinite trees,
whose satisfiability was proved decidable in~\cite{Bojanczyk04}.
This relies on an $\MSO$ interpretation of pushdown graphs into infinite trees, following~\cite{MullerSchupp85}.
More specifically, let $\G = (\A,\fin\parp)$ be a pushdown finitary parity game.
We construct the memory structure $\M$ as in Proposition~\ref{prop:mem_bounded_parity}, which keeps track
of the most important request.
The arena $\A \times \M$ is again a pushdown arena, so it can be $\MSO$-interpreted into the infinite binary tree.
Now thanks to Theorem~\ref{thm:strategy_complexity}, Eve has a memoryless winning strategy 
in $\G \times \M = (\A \times \M,\fin\parp)$ from her winning set.
Such a strategy can be described as a set of edges, hence as a monadic second-order variable
in an $\MSO$ formula over the infinite binary tree.
Consider the following formula:
$$
\exists X, \exists Y,\
\left \{
\begin{array}{cr}
X \textrm{ represents a positional strategy } & \wedge \\
\textrm{all infinite paths end up in } Y & \wedge \\
\Bn Z,\ Z \textrm{ path in } X \cap Y \wedge \textrm{ the minimal color in } Z \textrm{ is odd}
\end{array}
\right.
$$
It expresses the existence of a positional strategy ($X$), a subset of vertices ($Y$)
and a bound $N \in \N$
such that all plays consistent with $X$ eventually enter in $Y$,
where every request in answered within $N$ steps.
This is an existential bounding formula equivalent to the above property,
whose satisfiability is decidable~\cite{Bojanczyk04}.

\vskip1em
\noindent The second proof relies on~\cite{BlumensathColcombetKuperbergVandenboom13},
which studies two-way alternating parity cost-automata over infinite trees.
For our purpose, we consider such automata over $\Gamma$-trees,
which are infinite trees where each node has one child for each element in $\Gamma$.
(Later, $\Gamma$ will be the stack alphabet of a pushdown system.)
We denote by $\Act(\Gamma)$ the following set of actions on a $\Gamma$-tree: 
$\set{\push(a) \mid a \in \Gamma} \cup \set{\pop,\Skip}$,
where $\push(a)$ should be understood as ``going down in the direction $a$'', $\pop$ as ``going up'' and $\Skip$ as ``no move''.
(Note that they are in one-to-one correspondence with actions on a stack over $\Gamma$.)

\begin{definition}
A two-way alternating automaton over $\Gamma$-trees (with $k$ counters)
is a tuple $\B = (Q,A,\delta,q_0,c)$,
where $Q$ is a finite set of states, $A$ is a finite alphabet, 
$\delta : Q \times A \rightarrow \mathbb{B}^+(\Act(\Gamma) \times Q \times \set{\varepsilon,i,r}^k)$ is a transition relation
(note that the counter actions appear here),
$q_0$ is an initial state and $c : Q \rightarrow \N$ is a coloring function. 
\end{definition}

Let $\B$ be such an automaton, it can be considered under two semantics:
as an $\omega B$-automaton or as a parity cost-automaton.
Let $t$ be a $\Gamma$-tree; the automaton $\B$ and the tree $t$ induce 
an infinite-state game $(\A_{\B,t},\bounded \cap \parp)$.
As an $\omega B$-automaton, $t$ is accepted by $\B$ if Eve wins the $\omega B$ game,
and as a parity cost-automaton, $t$ is accepted by $\B$ if there exists $N \in \N$ such that
Eve wins the $\omega B$ game for the bound $N$.

The membership problem for such automata is a decision problem which asks,
given a two-way alternating parity cost-automaton $\B$ and a regular tree $t$,
whether $t$ is accepted by $\B$.
The following result is a consequence of~\cite{BlumensathColcombetKuperbergVandenboom13}:

\begin{theorem}[\cite{BlumensathColcombetKuperbergVandenboom13}]
The membership problem for two-way alternating parity cost-automata over regular trees is decidable.
\end{theorem}
Indeed, they prove that two-way alternating parity cost-automata can be effectively translated
into one-way alternating parity cost-automata, for which the membership problem
is known to be decidable.
We reduce the problem of solving a pushdown $\omega B$ game to the membership problem 
for two-way alternating parity cost-automata over a given regular tree.

\vskip1em
\noindent Following~\cite{KupfermanVardi00}, we first reduce
the problem of determining the winner in a pushdown game
to the membership problem for two-way alternating $\omega B$ automata
over regular trees.

Consider the pushdown $\omega B$ game $\G = (\A,\bounded \cap \parp)$, 
and fix an initial configuration $(q_0,\bot)$ for the game.
We define a two-way alternating $\omega B$ automaton:
let $\B = (Q,\Gamma_\bot,\delta,q_0,c)$,
where the transition relation $\delta$ is defined as follows:
$\delta(p,a)$ is the disjunction of all possible transitions from $(p,a)$ if $p \in Q_E$,
and the conjunction of all possible transitions from $(p,a)$ if $p \in Q_A$.

We run the automaton $\B$ on the $\Gamma$-tree $t$, which represents the stack contents:
the label of the node $u$ is the last letter of $u$ if $u \neq \varepsilon$, and $\bot$ otherwise.

\begin{lemma}[\cite{KupfermanVardi00}]
The following are equivalent:
\begin{itemize}
	\item Eve wins $\G$ from $(q_0,\bot)$.
	\item $t$ is accepted by $\B$.
	\item $\exists \sigma \textrm{ strategy for Eve}, \forall \pi \textrm{ plays}, \exists N \in \N,
	\pi \in \limitbounded(N) \cap \parp$.
\end{itemize}
\end{lemma}
Now Theorem~\ref{thm:collapse} implies that this is also equivalent to:
\[
\exists \sigma \textrm{ strategy for Eve}, \exists N \in \N, \forall \pi \textrm{ plays},\
\pi \in \limitbounded(N) \cap \parp\ .
\]
We construct a two-way alternating cost-automaton $\B'$ from $\B$ such that
$\B'$ accepts $t$ if and only if $\B$ accepts $t$ (as an $\omega B$-automaton).
The automaton $\B'$ is obtained by adding at each transition the ability to reset all counters
at the price of visiting a very bad color for the parity condition,
which takes care of the difference between $\limitbounded(N)$ and $\bounded(N)$.

The main result of this section follows:
\begin{theorem}
Solving a pushdown $\omega B$-game is decidable.
\end{theorem}

%% file: lower_bound.tex
In this subsection, we prove lower bounds on the collapse bound which appears in Theorem~\ref{thm:collapse},
focusing on pushdown finitary B\"uchi games; 
note that this implies the same lower bound for the more general case of pushdown $\omega B$ games.

For the special case of pushdown finitary B\"uchi games, Theorem~\ref{thm:collapse}
can be stated as follows:
\begin{corollary}
For all pushdown finitary B\"uchi games, there exists $N$ such that:
$$\WE(\fin\bucf) = \WE(\buc(F,N))\ .$$
\end{corollary}
The collapse bound depends on the following two relevant parameters of the pushdown arena: $n = |Q|$,
the number of states, and $k = |\Gamma|$, the size of the stack alphabet.

We show that the collapse bound is at least doubly-exponential in the number of states 
and exponential in the stack alphabet.

\subsubsection{The collapse bound for deterministic pushdown systems}

We start by considering deterministic pushdown systems,
which is the very restricted case of pushdown games where 
from every configuration, there is only one transition,
so no player has choice.

Standard pumping arguments shows that the collapse bound is at most exponential in both
the number of states and the stack alphabet.

\begin{lemma}\label{lem:pushdown_deterministic_upperbound}
For all deterministic pushdown systems,
we have:
$$\WE(\fin\bucf) = \WE(\buc(F,N))\ ,$$ 
for $N = n^2 \cdot k^{n \cdot k + 1}$.
\end{lemma} 

\begin{proof}
We prove the left-to-right inclusion.
Consider a path $\pi$, and assume it satisfies the finitary B\"uchi condition $\fin\bucf$.
We will show that it also satisfies the uniform B\"uchi condition for the bound $N = n^2 \cdot k^{n \cdot k + 1}$.
This collapse result is similar in fashion to the one obtained from the study of finitary games
over finite arenas. 
It is clear that in this setting, if a path in a \textit{deterministic} arena satisfies the finitary B\"uchi condition,
then it satisfies the uniform B\"uchi condition for the bound $n-1$ ($n$ being the number of vertices).
Indeed, such a path is ultimately periodic, and the simple cycle it describes has length at most $n$.
The content of this proof is to exhibit such a periodic pattern in $\pi$.
Using a case distinction, we prove that either $\pi$ ultimately repeats a cycle of length at most $n \cdot k^{n \cdot k}$,
or ultimately repeats a cycle of increasing height (with respect to the stack) 
of length at most $n^2 \cdot k^{n \cdot k + 1}$.
The two cases we consider are the following, they are illustrated in Figure~\ref{fig:pushdown_deterministic_upperbound}:
\begin{enumerate}
	\item there is some configuration that appears twice;
	\item no configuration appears twice.
\end{enumerate}

Before going through these two cases, we state an observation that will be used several times in the proof:
a simple path (that is, where each configuration appears at most once) 
whose maximal stack height difference is less than $H$
has length at most $n \cdot k^H$.

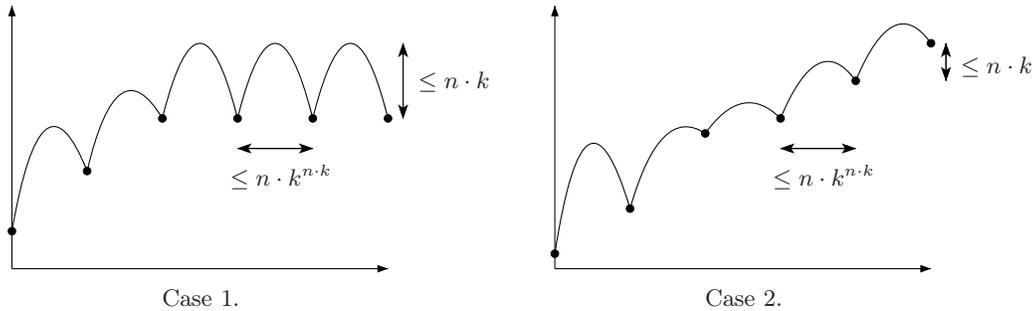
\begin{figure}
\begin{center}
\begin{picture}(60,40)(0,-2)
	\drawline(0,0)(50,0)
	\drawline(0,0)(0,35)

	\gasset{Nfill=y,Nh=1,Nw=1,AHLength=0}

	\node(1)(0,5){}
	\node(2)(10,13){}
	\node(3)(20,20){}
	\node(4)(30,20){}
	\node(5)(40,20){}
	\node(6)(50,20){}

	\drawqbezier(0,5,4,28,10,13)
	\drawqbezier(10,13,14,30,20,20)
	\drawedge[curvedepth=10](3,4){}
	\drawedge[curvedepth=10](4,5){}
	\drawedge[curvedepth=10](5,6){}

	\gasset{Nfill=y,Nh=1,Nw=1,AHLength=2,ATnb=1}
	\drawline(30,16)(40,16)
	\drawline(40,16)(30,16)
	\put(29,11){$\le n \cdot k^{n \cdot k}$}

	\drawline(52,20)(52,30)
	\drawline(52,30)(52,20)
	\put(54,24){$\le n \cdot k$}

	\put(20,-5){Case 1.}
\end{picture}
\hspace*{1cm}
\begin{picture}(60,40)(0,-2)
	\drawline(0,0)(50,0)
	\drawline(0,0)(0,35)

	\gasset{Nfill=y,Nh=1,Nw=1,AHLength=0}

	\node(1)(0,2){}
	\node(2)(10,8){}
	\node(3)(20,18){}
	\node(4)(30,20){}
	\node(5)(40,25){}
	\node(6)(50,30){}

	\drawqbezier(0,2,4,28,10,8)
	\drawqbezier(10,8,14,22,20,18)
	\drawqbezier(20,18,25,25,30,20)
	\drawqbezier(30,20,35,32,40,25)
	\drawqbezier(40,25,45,37,50,30)

	\gasset{Nfill=y,Nh=1,Nw=1,AHLength=2,ATnb=1}
	\drawline(30,16)(40,16)
	\drawline(40,16)(30,16)
	\put(29,11){$\le n \cdot k^{n \cdot k}$}

	\drawline(52,25)(52,30)
	\drawline(52,30)(52,25)
	\put(54,26){$\le n \cdot k$}

	\put(20,-5){Case 2.}
\end{picture}
\end{center}
\caption{The case distinction for Lemma~\ref{lem:pushdown_deterministic_upperbound}.}
\label{fig:pushdown_deterministic_upperbound}
\end{figure}

We start with the first case. It is clear that $\pi$ is ultimately periodic; let $C$ be the simple cycle described by $\pi$.
We can see that the maximal stack height difference in the cycle $C$ is less than $n \cdot k$,
relying on a vertical pumping argument.
It follows, relying on the earlier observation, that the cycle has length at most $n \cdot k^{n \cdot k}$.

We now focus on the second case. Here $\pi$ is not ultimately periodic, 
but we will show that it repeats a cycle of increasing height.
Define a step to be a configuration $(q,u \bot)$ in $\pi$ whose stack height is minimal among the configurations
that are visited after $(q, u \bot)$ in $\pi$.
Since no configuration appears twice, it is clear that $\pi$ has infinitely many steps.
We say that two steps are consecutive in $\pi$ if there are no steps inbetween in $\pi$.
We first observe that two consecutive steps are separated by at most $n \cdot k^{n \cdot k}$ transitions:
indeed the stack height, which remains higher than the height of the first step, 
must remain within the $n \cdot k$ intervall above the first step.
Consider now the $n \cdot k$ first steps; two of them share the same state and top stack content,
let us denote them $(q,au \bot)$ and $(q,avau \bot)$.
The path $\pi$ ultimately repeats a cycle of increasing height, as follows:
$$(q,au \bot) \rightarrow (q,avau \bot) \rightarrow \ldots \rightarrow (q,(av)^pau \bot),$$
whose length is bounded by $n^2 \cdot k^{n \cdot k + 1}$.
This concludes.
\hfill\qed
\end{proof}

The collapse bound proved in this lemma seems a priori quite large for such an easy case,
as it is exponential in both $n$ and $k$.
However, Example~\ref{ex:pushdown_deterministic_lowerbound} shows
that it is asymptotically tight.

\begin{example}
\label{ex:pushdown_deterministic_lowerbound}
Figure~\ref{fig:pushdown_deterministic_lowerbound} presents a deterministic pushdown system,
where the only path from $(F,\bot)$ satisfies the condition $\buc(F,N)$ for $N = O(2^n)$
but not for asymptotically less.
This system encodes a number in binary in the stack with the least significant bit on the top of the stack.
It has two phases: an initialization phase and an increment phase.

The initialization phase has $n$ states and consists in pushing $n$ times the symbol $0$.
The increment phase consists in adding one to the number encoded in the stack,
\textit{i.e} $1^k0u\bot \xrightarrow{\ *\ } 0^k1u\bot$.
This phase goes on until it reaches the stack content $1^n \bot$, which is emptied to reach the only B\"uchi state $F$,
and restarts from scratch.
This pushdown process has $O(n)$ states and the collapse bound is $O(2^n)$.

An easy generalization consists in encoding in base $k$ instead of $2$, 
which would give an arena of size $O(k \cdot n)$ and a collapse bound asymptotically in $k^n$,
\textit{i.e} exponential in the number of states but not in the stack alphabet.

To obtain an arena where the collapse bound is exponential in both parameters,
we perform slight modifications, as follows.
In the latter arena, the numbers are encoded with $n$ bits; we improve this by encoding the numbers using $k \cdot n$ bits.
The increment phase remains the same.
The initialization phase is not optimal;
an ideal initialization phase would use $O(n)$ states to push $0^{k \cdot n}$ on the stack,
but this is not possible, so we use a weaker initialization phase with $n$ states that pushes:
$$\underbrace{(k-1) \ldots (k-1)}_n \cdot \underbrace{(k-2) \ldots (k-2)}_n \cdot \ldots \cdot 
\underbrace{1 \ldots 1}_n \cdot \underbrace{0 \ldots 0}_n \ .$$
The modified gadget is represented in Figure~\ref{ex:pushdown_deterministic_improved_gadget}.

Since the counter does not start from $0$ but from the number encoded in the latter stack,
this new arena performs at bit less than $k^{k \cdot n}$ increment phases,
but more than half this number, so its collapse bound is $O(k^{k \cdot n})$,
exponential in both $n$ and $k$.

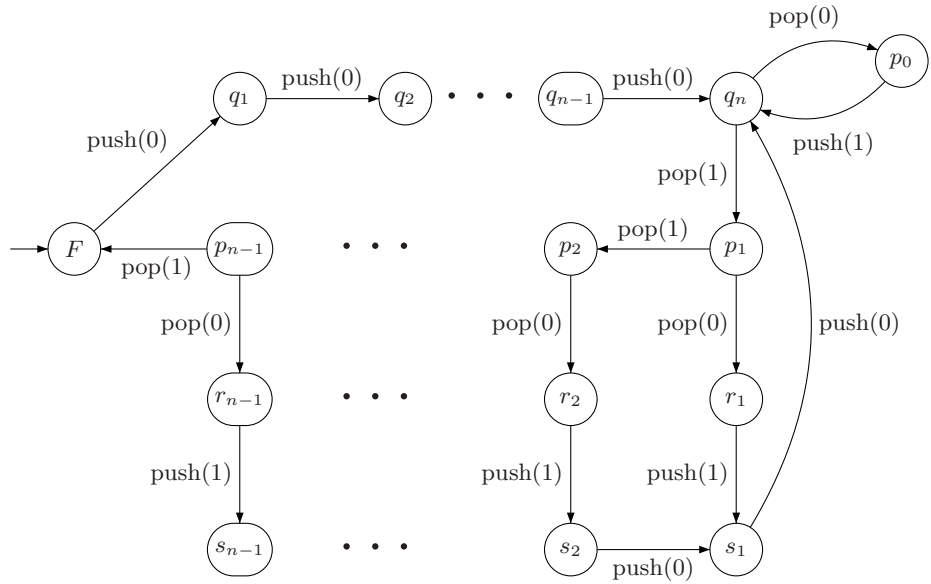
\begin{figure}
\begin{center}
\begin{picture}(110,73)(0,0)
	\gasset{Nh=7,Nw=7}

	\node[Nmarks=i,polyangle=45](F)(0,40){$F$}

	\node[polyangle=45](q_1)(22,60){$q_1$}
	\node[polyangle=45](q_2)(44,60){$q_2$}
	\put(49,60){\begin{Huge}$\ldots$\end{Huge}}
	\node[polyangle=45,Nadjust=w](q_n-1)(66,60){$q_{n-1}$}
	\node[polyangle=45](q_n)(88,60){$q_n$}

	\node[polyangle=45](p_0)(110,65){$p_0$}

	\node[polyangle=45,Nadjust=w](p_n-1)(22,40){$p_{n-1}$}
	\put(35,40){\begin{Huge}$\ldots$\end{Huge}}
	\node[polyangle=45](p_2)(66,40){$p_2$}
	\node[polyangle=45](p_1)(88,40){$p_1$}

	\node[polyangle=45,Nadjust=w](r_n-1)(22,20){$r_{n-1}$}
	\put(35,20){\begin{Huge}$\ldots$\end{Huge}}
	\node[polyangle=45](r_2)(66,20){$r_2$}
	\node[polyangle=45](r_1)(88,20){$r_1$}

	\node[polyangle=45,Nadjust=w](s_n-1)(22,0){$s_{n-1}$}
	\put(35,0){\begin{Huge}$\ldots$\end{Huge}}
	\node[polyangle=45](s_2)(66,0){$s_2$}
	\node[polyangle=45](s_1)(88,0){$s_1$}

	\drawedge(F,q_1){$\push(0)$}
	\drawedge(q_1,q_2){$\push(0)$}
	\drawedge(q_n-1,q_n){$\push(0)$}

	\drawedge[curvedepth=5](q_n,p_0){$\pop(0)$}
	\drawedge[curvedepth=5](p_0,q_n){$\push(1)$}
	
	\drawedge[ELside=r](q_n,p_1){$\pop(1)$}
	\drawedge[ELside=r](p_1,p_2){$\pop(1)$}
	\drawedge(p_n-1,F){$\pop(1)$}

	\drawedge[ELside=r](p_1,r_1){$\pop(0)$}
	\drawedge[ELside=r](p_2,r_2){$\pop(0)$}
	\drawedge[ELside=r](p_n-1,r_n-1){$\pop(0)$}

	\drawedge[ELside=r](r_1,s_1){$\push(1)$}
	\drawedge[ELside=r](r_2,s_2){$\push(1)$}
	\drawedge[ELside=r](r_n-1,s_n-1){$\push(1)$}

	\drawedge[ELside=r](s_2,s_1){$\push(0)$}
	\drawedge[ELside=r,curvedepth=-10](s_1,q_n){$\push(0)$}
\end{picture}
\end{center}
\caption{A deterministic pushdown system with an exponential collapse bound.}
\label{fig:pushdown_deterministic_lowerbound}
\end{figure}
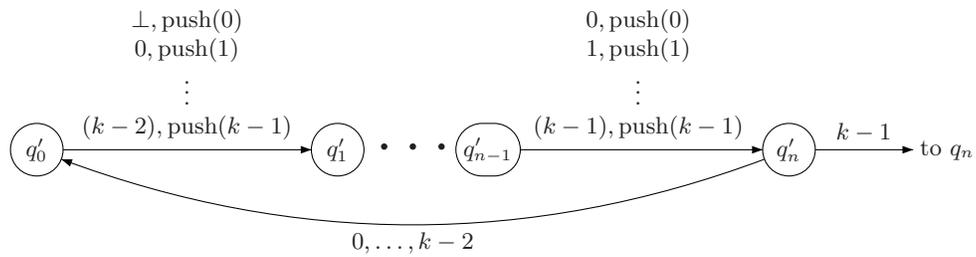
\begin{figure}
\begin{center}
\begin{picture}(120,35)(0,0)
	\gasset{Nh=7,Nw=7}

	\node[polyangle=45](q_0)(0,10){$q'_0$}
	\node[polyangle=45](q_1)(40,10){$q'_1$}
	\put(45,10){\begin{Huge}$\ldots$\end{Huge}}
	\node[polyangle=45,Nadjust=w](q_n-1)(60,10){$q'_{n-1}$}
	\node[polyangle=45](q_n)(100,10){$q'_n$}
	\node[polyangle=45,Nframe=n](q_G)(120,10){\ \ to $q_n$}

	\drawedge(q_0,q_1){$\begin{array}{cc}\bot,\push(0) \\ 0,\push(1) \\ \vdots \\ (k-2),\push(k-1)\end{array}$}
	\drawedge(q_n-1,q_n){$\begin{array}{cc}0,\push(0) \\ 1,\push(1) \\ \vdots \\ (k-1),\push(k-1)\end{array}$}
	\drawedge[curvedepth=10](q_n,q_0){$0,\ldots,k-2$}
	\drawedge(q_n,q_G){$k-1$}

\end{picture}
\end{center}
\caption{The improved initialization gadget.}
\label{ex:pushdown_deterministic_improved_gadget}
\end{figure}
\end{example}

\subsubsection{The collapse bound for pushdown games}

For the following three examples, we denote by $p_1,p_2,\ldots$ the sequence of prime numbers,
and by $q_n$ the product of the first $n$ prime numbers.
We first start with the case where the stack alphabet has size one,
\textit{i.e} the subclass of one-counter pushdown games.
Example~\ref{ex:lower_bound_one_counter} shows that in this case
the bound is exponential in the number of states.

\begin{example}
\label{ex:lower_bound_one_counter}
Figure~\ref{fig:lower_bound_one_counter} presents a one-counter pushdown game,
where for the condition $\buc(F,N)$, Eve wins for $N = q_n$
but not for $N-1$.
Eve first pushes a sequence of $a$'s on the stack, 
then Adam chooses a prime number up to $p_n$ 
and checks that the size of this sequence is divisible by this number.
For this, Adam goes to a loop of size $p_k$,
going deterministically through it while popping one $a$ at a time.
If the empty stack is encountered in the beginning of the loop,
then the size of the stack is divisible by $p_k$,
and the game starts from scratch, visiting a B\"uchi state on the way.

Since Eve does not know in advance which prime number Adam is going to choose among $p_1,\ldots,p_n$,
she has to push a non-empty sequence of size divisible by $q_n = \Pi_{1 \leq i \leq n} p_i$.
The size of the arena is $O(\sum_{1 \leq i \leq n} p_i) = O(n \cdot p_n)$, 
whereas the smallest bound Eve can secure is $q_n$.
An easy calculation shows that $q_n$ is exponential in $O(n \cdot p_n)$.
\begin{figure}
\begin{center}
\begin{picture}(90,50)(-3,0)
	\node[Nmarks=i,polyangle=45](push)(0,20){$i$}
	\drawloop[loopangle=90](push){$\push(a)$}

	\rpnode[polyangle=45](c)(15,20)(4,4){}

	\rpnode[polyangle=45](12)(35,40)(4,4){$1$}
	\rpnode[polyangle=45](02)(55,40)(4,4){$0$}
	
	\drawedge(push,c){$a$}

	\drawedge(c,12){$\pop(a)$}
	\drawedge[curvedepth=5](12,02){$\pop(a)$}
	\drawedge[curvedepth=5](02,12){$\pop(a)$}

	\rpnode[polyangle=45](13)(35,10)(4,4){$1$}
	\rpnode[polyangle=45](03)(55,20)(4,4){$0$}
	\rpnode[polyangle=45](23)(55,0)(4,4){$2$}
 	
	\drawedge[ELside=r](c,13){$\pop(a)$}
	\drawedge[ELside=r](23,03){$\pop(a)$}
	\drawedge[ELside=r](03,13){$\pop(a)$}
	\drawedge[ELside=r](13,23){$\pop(a)$}

  	\node(invisible)(75,20){$F$}
  	\node[linecolor=White](toi)(90,20){\large{\ to $i$}}
  	\drawedge(02,invisible){$\bot$}
  	\drawedge(03,invisible){$\bot$}
  	\drawedge(invisible,toi){}

%	\put(43,-3){\begin{huge}$\vdots$\end{huge}}
\end{picture}
\end{center}
\caption{A one-counter pushdown game with an exponential bound.}
\label{fig:lower_bound_one_counter}
\end{figure}
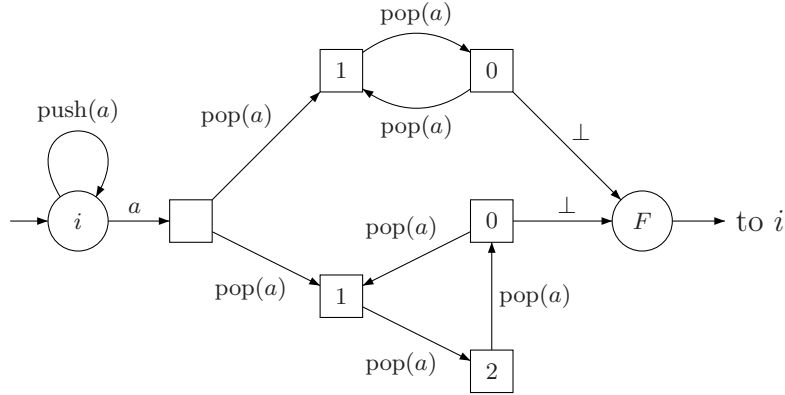
\end{example}

We now consider a stack alphabet of size two, 
and combine the two ideas underlying Example~\ref{ex:pushdown_deterministic_lowerbound}
and Example~\ref{ex:lower_bound_one_counter}, that is:
\begin{itemize}
	\item Eve needs to push a sequence of exponential size;
	\item this sequence, seen as a binary decomposition of the number $0$, is incremented by one
until it reaches the sequence of only $1$'s, where the game empties the stack, 
starts from scratch and visits a B\"uchi state along the way.
\end{itemize}
Example~\ref{ex:double_exponential_lower_bound} implements this idea,
showing that the collapse bound is at least doubly-exponential in the number of states.

\begin{example}
\label{ex:double_exponential_lower_bound}
Figure~\ref{fig:double_exponential_lower_bound} presents a pushdown game,
where Eve wins for the condition $\buc(F,N)$ for $N = O(2^{q_n})$,
but not for $N = o(2^{q_n})$.
In the figure, ``sh'' stands for stack-height: we saw in Example~\ref{ex:lower_bound_one_counter}
how Adam can check that the size of the stack is a multiple of $q_n$, product
of the $n$ first prime numbers, using only $O(n \cdot p_n)$ states.
As in the previous example, Eve first pushes a sequence of $0$'s on the stack, 
whose length must be a multiple of $q_n$,
otherwise Adam wins by checking it.
From $s$ starts a binary increment similar to the one presented in Example~\ref{ex:pushdown_deterministic_lowerbound};
however in this example, the number of bits allowed was linear in the size of the arena,
and we are now lifting this up to an exponential number of bits.
So, we have to rely on the players' interactions to ensure that the binary increment
is correctly executed.
The action performed in the stack should be:
$$(s,1^k0u\bot) \xrightarrow{\ *} (s,0^k1u\bot).$$
The first part is deterministic:
$$(s,1^k0u\bot) \xrightarrow{\ *} (c,1u\bot).$$
From $c$, Eve pushes some $0$ on the stack.
If she pushes less than $k$ symbols, then Adam wins by checking,
so she has to push at least $k$.
Note, however, that she could push $k$ plus any multiple of $q_n$,
but she would only do herself a disservice.

The arena has size $O(n \cdot p_n)$, so $N$ is doubly-exponential in the number of states.

\begin{figure}
\begin{center}
\begin{picture}(100,25)(-3,0)
	\gasset{Nw=6,Nh=6,loopdiam=5}
	
	\node[Nmarks=i,polyangle=45](push)(0,3){$i$}
	\rpnode[polyangle=45](choix)(15,3)(4,4){}
  	\node[Nadjust=w](check)(15,23){sh $\not\equiv 0 [q]$}
	\node(s)(30,3){$s$}
	\node(int)(50,3){}
	\node(c)(70,3){$c$}
  	\node(F)(85,3){$F$}
  	\node[linecolor=White](to_i)(100,3){\large{\ to $i$}}
 	
	\drawedge(choix,s){}
	\drawedge(choix,check){}
	\drawloop[loopangle=90](push){$\push(0)$}	
	\drawedge(push,choix){$0$}
	\drawloop[loopangle=90](s){$\pop(1)$}	
	\drawedge(s,int){$\pop(0)$}
	\drawedge(int,c){$\push(1)$}
	\drawloop[loopangle=90](c){$\push(0)$}	
	\drawedge[curvedepth=7](c,choix){}

  	\drawedge(c,F){$\bot$}
  	\drawedge(F,to_i){}
\end{picture}
\end{center}
\caption{A pushdown game with a doubly exponential bound.}
\label{fig:double_exponential_lower_bound}
\end{figure}
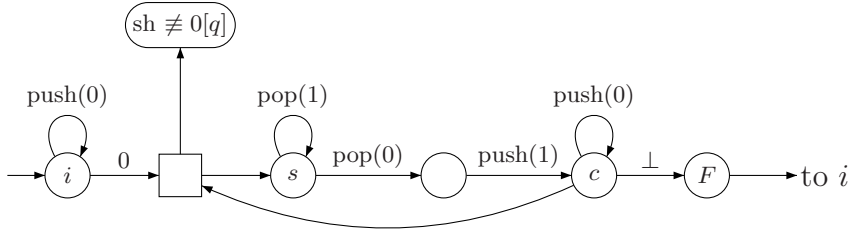
\end{example}

We now turn to a stack alphabet of size $2k+1$, 
and roughly ``nest'' Example~\ref{ex:double_exponential_lower_bound}.

Let $\Gamma = \set{a_1,b_1,\ldots,a_k,b_k} \cup \set{\sharp}$.
The stack configurations we consider belong to the regular language:
$$L = \bigcup_{1 \leq i \leq k} 
\left(\set{a_i,b_i}^{q_n}\right)^+ \cdot \sharp \cdot 
\left(\set{a_{i-1},b_{i-1}}^{q_n}\right)^+ \cdot \sharp 
\ldots \sharp \cdot \left(\set{a_1,b_1}^{q_n}\right)^+.$$
Each block $\left(\set{a_i,b_i}^{q_n}\right)^+$ is seen as a number encoded in binary,
where $a_i$ is $0$ and $b_i$ is $1$,
which is initialized to $a_i^{q_n}$ and incremented by one step by step.
However, the incrementation policy requires that to increment in the $i$\textsuperscript{th}
block for $i < k$, one must increment in the $(i+1)$\textsuperscript{th} block.
Hence two increment phases in the $i$\textsuperscript{th} block
are separated by $2^{q_n}$ increment phases in the $(i+1)$\textsuperscript{th} block,
which implies that two increment phases in the first block
are separated by $2^{(k-1) \cdot q_n}$ transitions.
Hence the $2^{q_n}$ increment phases required in the first block
are executed within $2^{k \cdot q_n}$ steps.
Example~\ref{ex:stack_alphabet_lower_bound} constructs such a game.

\begin{example}
\label{ex:stack_alphabet_lower_bound}
We sketch the construction of a pushdown game,
where Eve wins $\buc(F,N)$ for $N = O(2^{k \cdot q_n})$, but not for asymptotically less.

First, following an easy adaptation of Example~\ref{ex:lower_bound_one_counter}
we construct a game where Eve wins if and only if the stack content belongs to the language $L$.
It has $k$ components, each in charge of checking a block $\set{a_i,b_i}^{q_n}$.
Eve first chooses $i$, and then Adam chooses a prime number to check that the size of the block
is a multiple of the chosen prime number.
Once a $\sharp$ symbol is reached, it is popped and the run goes on 
with the $(i-1)$\textsuperscript{th} component, until the stack is empty.
The size of this game is $O(k \cdot n \cdot p_n)$.

As before, Eve first pushes a sequence of $a_1$'s on the stack, 
whose length must be a multiple of $q_n$,
otherwise Adam wins by checking it.
If he sends the pebble to $d$, then Eve chooses an $i$
and starts a binary increment from $v_i$, 
similar to the one presented in Example~\ref{ex:double_exponential_lower_bound}.
There are some differences, which appear at the end of an increment phase.
If the block contained no $a_i$'s, then the following case distinction occurs:
\begin{itemize}
	\item If $1 < i \leq k$,
then the symbol $\sharp$ is popped from the stack, and another increment phase starts from $v_{i-1}$.
	\item If $i = 1$,
then the game starts from scratch after paying a visit to a B\"uchi state.
\end{itemize}
Otherwise, the first $a_i$ is turned into a $b_i$, and then Eve pushes some $a_i$'s
before sending the pebble to a state controlled by Adam.
There, he can check that the stack content belongs to $L$,
but he also has another option, following the case distinction:
\begin{itemize}
	\item If $1 \leq i < k$,
then Adam can send the pebble back to the initial state, pushing a $\sharp$ symbol along the way.
	\item If $i = k$, 
then Adam can send the pebble to $v_k$.
\end{itemize}

Whenever Adam sends the pebble back to the initial state 
after an increment phase of the $i$\textsuperscript{th} block,
Eve has no choice but to push a sequence of $a_{i+1}$'s on the stack,
whose length must be a multiple of $q_n$,
otherwise Adam wins since the stack content would not belong to $L$.

The arena obtained has size $O(k \cdot n \cdot p_n) + O(k) = O(k \cdot n \cdot p_n)$,
so the bound required for Eve to win the uniform B\"uchi condition
is doubly-exponential in the number of states and exponential in the stack alphabet.
\end{example}

%% file: reduction.tex
The reduction relies on a \emph{restart} gadget.
We consider a pushdown game with finitary parity conditions, given by the coloring function
$c : Q \rightarrow [d]$, where we assume $d$ to be odd.
Between every edge of the game we add a restart gadget, where Eve
can choose either to follow the edge, or to ``restart'':
this entails that first a vertex with priority $0$ is visited, 
where Adam can stay as long as he wants by pushing on the stack a new symbol $\sharp$,
and then Eve takes over, staying in a vertex with priority $d$ 
until all the $\sharp$ symbols are popped away from the stack,
before following the original edge.
The intuition is the following: whenever Eve chooses to restart, visiting the vertex with priority $0$
answers all previous requests, but this comes with the cost that Adam
will be able to let a request unanswered for a long time.
Therefore, Eve can restart only finitely many times.
The gadget is represented in Figure~\ref{fig:restart_gadget}.

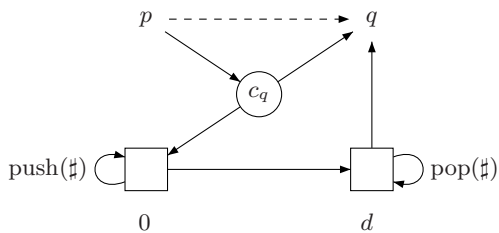
\begin{figure}
\begin{center}
\begin{picture}(30,30)(0,-5)
	\gasset{Nw=6,Nh=6,loopdiam=4}

	\node[linecolor=White](p)(0,20){$p$}
	\node[linecolor=White](q)(30,20){$q$}
	\node(c)(15,10){$c_q$}
	\rpnode[polyangle=45](0)(0,0)(4,4){}
	\rpnode[polyangle=45](2d+1)(30,0)(4,4){}

	\put(-1,-8){$0$}
	\put(28.5,-8){$d$}
		
	\drawedge[dash={1}0](p,q){}
	\drawedge(p,c){}
	\drawedge(c,q){}
	\drawedge(c,0){}
	\drawloop[loopangle=180](0){$\push(\sharp)$}
	\drawloop[loopangle=0](2d+1){$\pop(\sharp)$}
	\drawedge(0,2d+1){}
	\drawedge(2d+1,q){}
\end{picture}
\end{center}
\caption{The restart gadget.}
\label{fig:restart_gadget}
\end{figure}

\begin{lemma}
Eve wins the finitary parity game if and only if 
she wins the reduced bounded parity game.
\end{lemma}

\begin{proof}
We prove both implications.
\begin{itemize}
	\item Assume Eve wins the finitary parity game, and let $\sigma$ be a winning strategy.
We construct a strategy $\sigma_R$ in the reduced bounded parity game. 
It maintains a counter, initially set to $1$, whose value is denoted by $N$.
The strategy $\sigma_R$ plays consistently with $\sigma$.
It restarts if there exists a request made before the last $N$ transitions that has not been serviced, 
and if so increments the counter by one.
We argue that $\sigma_R$ is winning for the bounded parity condition.
Consider $\pi_R$ a play consistent with $\sigma_R$:
if it remains in the restart gadget forever (Adam pushes $\sharp$ forever), it is winning.
Otherwise, if a restart occurs for a value $N$ of the counter, 
then there is a pending request not serviced within $N$
transitions, which got serviced through the restart.
Let $\pi$ be the corresponding play in the parity game, 
where we skip the restarts: $\pi$ is consistent with $\sigma$,
so it satisfies the finitary parity condition.
Now, it is clear that $\pi_R$ contains only finitely many restarts, 
otherwise it would include requests that are not serviced within $N$ transitions, 
for arbitrary $N$,
which contradicts the fact that $\pi$ satisfies the finitary parity condition.
It follows that $\pi_R$ and $\pi$ coincide from some point onwards,
so $\pi_R$ satisfies the bounded parity condition,
and $\sigma_R$ is a winning strategy in the reduced bounded parity game.
	\item Conversely, assume that Adam wins the finitary parity game, and let $\tau$ be a winning strategy.
We construct a strategy $\tau_R$ in the reduced bounded parity game.
As for the case of Eve, it features a counter, initialized to $1$ and whose value is denoted by $N$.
Outside the restart gadget, $\tau_R$ plays consistently with $\tau$,
and inside the restart gadget, $\tau_R$ pushes exactly $N$ times the symbol $\sharp$,
and then increments the counter by one.
Consider $\pi_R$ a play consistent with $\tau_R$, there are two cases:
either it includes finitely many uses of the restart gadgets, 
or infinitely many.
In the first case, $\pi_R$ coincides from some point onwards with a play $\pi$ consistent with $\tau$, 
so it spoils the bounded parity condition.
In the second case, the request made in the last vertex of the restart gadget remains
unserviced for an unbounded time, so the bounded parity condition is fooled as well.
It follows that $\pi_R$ spoils the bounded parity condition,
thus $\tau_R$ is a winning strategy in the reduced bounded parity game.
\end{itemize}
\hfill\qed
\end{proof}

%% file: buchi_stack_boundedness.tex
In the study of finitary games over finite graphs~\cite{ChatterjeeHenzingerHorn09},
the following observation is made: finitary B\"uchi coincide with B\"uchi,
while finitary parity differs from parity as soon as three colors are involved.
Over pushdown arenas, even finitary B\"uchi differs from B\"uchi, as noted in Example~\ref{ex:intro}.
Yet when intersected with the stack boundedness condition,
the case of finitary B\"uchi specializes again and collapses to B\"uchi.

\begin{lemma}
\label{lem:buchi_collapse}
For all pushdown games,
$$\WE(\fin\bucf \cap \SB) = \WE(\bucf \cap \SB)\ .$$
\end{lemma}

The left-to-right inclusion is clear, since $\fin\bucf \subset \bucf$.
The converse inclusion follows from memoryless determinacy 
for the condition $\bucf \cap \SB$~\cite{BouquetSerreWalukiewicz03}:
assume $\sigma$ is a memoryless strategy ensuring $\bucf \cap \SB$,
and let $\pi$ be a play consistent with $\sigma$.
First note that between two visits of the same configuration,
there must be a B\"uchi configuration,
otherwise iterating this loop would be a play consistent with $\sigma$ yet losing.
The second observation is that since the stack height remains smaller than a bound $N$,
the number of different configurations visited in $\pi$ is finite and bounded by a function of $N$.
The combination of these two arguments imply that $\pi$ satisfies $\fin\bucf$.

Note however that in general, for a pushdown game,
$\WE(\fin\parp \cap \SB) \neq \WE(\parp \cap \SB)$.

%% file: parity_stack_boundedness.tex
We show how to use both ideas to handle pushdown games with finitary parity
and stack boundedness conditions.
We present a three-step reduction, illustrated in Figure~\ref{fig:reductions}.

\begin{figure}[!ht]
\begin{center}
\begin{picture}(75,18)(0,0)
	\gasset{Nadjust=wh,Nframe=n,AHLength=2,AHlength=3,curvedepth=10}

	\node(finpar_cap_sb)(0,4){\begin{tabular}{ccc}$\fin\parp$\\ \textrm{ and }\\$\SB$\end{tabular}}
	\node(bndpar_cap_sbq)(25,4){\begin{tabular}{ccc}$\bnd\parp$\\ \textrm{ and }\\$\SB(Q)$\end{tabular}}
	\node(finbuc_cap_sbq)(50,4){\begin{tabular}{ccc}$\fin\bucf$\\ \textrm{ and }\\$\SB(Q)$\end{tabular}}
	\node(buc_cap_sbq)(75,4){\begin{tabular}{ccc}$\bucf$\\ \textrm{ and }\\$\SB(Q)$\end{tabular}}

  	\drawedge(finpar_cap_sb,bndpar_cap_sbq){\textrm{restart}}
  	\drawedge(bndpar_cap_sbq,finbuc_cap_sbq){$\times \M$}
  	\drawedge(finbuc_cap_sbq,buc_cap_sbq){\textrm{collapse}}
\end{picture}
\caption{Sequence of reductions}
\label{fig:reductions}
\end{center}
\end{figure}
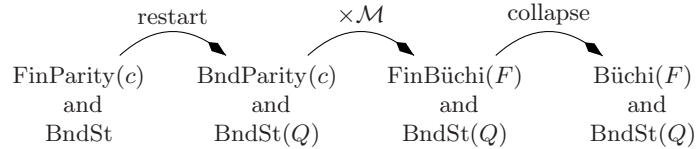

The first step is to adapt the reduction from finitary parity to bounded parity,
now intersected with the stack boundedness condition.
To this end, we need to modify the stack boundedness condition
so that it ignores the configurations in the restart gadget;
we define its restriction to $Q$:
$$\SB(Q) = \set{\pi \mid \exists N, 
\begin{array}{c}
\textrm{ all configurations in } \pi \\
\textrm{ with control state in } Q\\
\textrm{ have stack height less than } N
\end{array}
}\ .$$
Now the reduction is from finitary parity and stack boundedness to bounded parity and restricted stack boundedness.

The second step is the reduction from bounded parity to finitary B\"uchi
by composing with the memory structure from Proposition~\ref{prop:mem_bounded_parity},
keeping track of the most urgent pending request.
We are now left with a pushdown game with the condition finitary B\"uchi and restricted stack boundedness.

The third step is the collapse of finitary B\"uchi to B\"uchi.
Note that the collapse stated in Lemma~\ref{lem:buchi_collapse} deals
with stack boundedness, not restricted to a subset $Q$.
Indeed, the result does not hold in general for this modified stack boundedness condition,
but it does hold here due to the special form of the restart gadget,
that can be used only finitely many times.

Formally, we first need to extend the memoryless determinacy for the condition
B\"uchi and restricted stack boundedness.

\begin{lemma}
For all pushdown games with condition B\"uchi and restricted stack boundedness,
Eve has a memoryless winning strategy from her winning set.
\end{lemma}

\begin{proof}
The proof is a straightforward adaptation of Proposition~1 from~\cite{Gimbert04}.
\hfill\qed
\end{proof}

Now, consider $\sigma$ a memoryless strategy ensuring
the condition B\"uchi and restricted stack boundedness
in the pushdown game obtained through the above reductions;
we prove that $\sigma$ ensures finitary B\"uchi.
Let $\pi$ be a play consistent with $\sigma$,
there are two cases: either the play remains forever
in the restart gadget, or from some point onwards 
the restart gadget is not used anymore.
In the first case, the finitary B\"uchi condition is clearly satisfied.
In the other case, the play satisfies the general stack boundedness condition,
and the same reasoning as for Lemma~\ref{lem:buchi_collapse}
concludes that the finitary B\"uchi condition is satisfied.

\vskip1em
This three-step reduction produces in linear time
an equivalent pushdown game with the condition B\"uchi and stack boundedness restricted to $Q$.
It has been shown in~\cite{BouquetSerreWalukiewicz03,Gimbert04} that 
deciding the winner in a pushdown game with condition B\"uchi and stack boundedness
is $\EXPTIME$-complete;
a slight modification of their techniques extends this 
to the restricted definition of stack boundedness.
%\mynote{EXPLAIN}

\begin{theorem}
Determining the winner in a pushdown game with finitary parity and stack boundedness conditions
is $\EXPTIME$-complete.
\end{theorem}